\documentclass[a4paper,UKenglish,cleveref, autoref, thm-restate]{lipics-v2021}

\nolinenumbers
\hypersetup{colorlinks=true}

\title{Connected Dominating Set on Semi-Ladder-Free Graphs} 

\author{Sobyasachi Chatterjee}{The Institute of Mathematical Sciences, HBNI, Chennai, India}{sobyasachic@imsc.res.in}{https://orcid.org/0009-0002-5878-9975}{}
\author{Sushmita Gupta}{The Institute of Mathematical Sciences, HBNI, Chennai, India}{sushmitagupta@imsc.res.in}{https://orcid.org/0000-0003-1255-8266}{}
\author{Saket Saurabh}{The Institute of Mathematical Sciences, HBNI, Chennai, India \\ University of Bergen, Norway}{saket@imsc.res.in}{https://orcid.org/0000-0001-7847-6402}{}
\author{Sanjay Seetharaman}{The Institute of Mathematical Sciences, HBNI, Chennai, India}{sanjays@imsc.res.in}{https://orcid.org/0009-0001-1483-6138}{}
\author{Anannya Upasana}{The Institute of Mathematical Sciences, HBNI, Chennai, India}{anannyaupas@imsc.res.in}{https://orcid.org/0009-0002-6283-0846}{}

\authorrunning{S. Chatterjee et al.} 

\Copyright{Sobyasachi Chatterjee, Sushmita Gupta, Saket Saurabh, Sanjay Seetharaman, and Anannya Upasana} 



\category{} 

\relatedversion{} 

\funding{Sushmita Gupta acknowledges support from Anusandhan National Research Foundation (ANRF), grant number ANRF/ARGM/2025/003188/MTR.}

\EventEditors{John Q. Open and Joan R. Access}
\EventNoEds{2}
\EventLongTitle{42nd Conference on Very Important Topics (CVIT 2016)}
\EventShortTitle{CVIT 2016}
\EventAcronym{CVIT}
\EventYear{2016}
\EventDate{December 24--27, 2016}
\EventLocation{Little Whinging, United Kingdom}
\EventLogo{}
\SeriesVolume{42}
\ArticleNo{23}
\usepackage{verbatim}
\usepackage{setspace}
\usepackage[noend]{algpseudocode}
\usepackage{algorithm}
\usepackage[inline]{enumitem}

\newcommand{\app}{$\dagger$}

\def\version{compressed} 

\def\ShowComments{false}

\usepackage{amssymb,amstext,amsmath}

\usepackage[normalem]{ulem}
\usepackage{mathrsfs} 
\usepackage{xspace}
\usepackage{diagbox}
\usepackage{amsfonts}
\usepackage{pdfpages}
\usepackage{graphicx}

\usepackage{enumerate}
\usepackage{enumitem}
\usepackage{nicefrac}

\usepackage{comment}
\excludecomment{dontshow}

\usepackage[textsize=small]{todonotes}

\usepackage{mathtools}
\usepackage{ifthen} 
\usepackage{tabularx}
\usepackage{multirow}
\usepackage{booktabs}
\usepackage{array}
\usepackage{tcolorbox}

\usepackage{rotating}

\usepackage{multicol}
\usepackage{array}

\usepackage{colortbl}
\usepackage{tabulary}
\newcolumntype{K}[1]{>{\centering\arraybackslash}p{#1}}
\usepackage[dvipsnames, svgnames, x11names]{xcolor}
\usepackage{mdframed}
\usepackage{subcaption}
\usepackage{thm-restate}

\newcommand{\fpt}{{\sf FPT}\xspace}

\newcommand{\cF}{{\mathcal{F}}}
\newcommand{\cU}{\mathcal{U}}
\newcommand{\Oh}{\ensuremath{\mathcal{O}}\xspace}

\newcommand{\Co}[1]{\ensuremath{\mathcal{#1}}\xspace}
\newcommand{\X}[1]{\ensuremath{\mathscr{#1}}\xspace}

\newcommand{\yes}{{\sf YES}\xspace}
\newcommand{\no}{{\sf NO}\xspace}

\newcommand{\sse}{\subseteq}

\ifthenelse{\equal{\ShowComments}{true}}{
    \newcommand{\il}[1]{\todo[inline, color=yellow!40]{#1}}
    \newcommand{\ma}[1]{\todo[color=SpringGreen]{\small #1}}

    \newcommand{\Ma}[1]{\textcolor{magenta}{#1}}

}{
    \presetkeys{todonotes}{disable}{}

    \newcommand{\il}[1]{\hide{#1}}
    \newcommand{\ma}[1]{\hide{#1}}

    \newcommand{\Ma}[1]{#1}

}

\newcommand{\hide}[1]{}

\newcommand{\calG}{{\mathcal{G}}\xspace}

\newcommand{\defparprob}[4]{
\begin{tcolorbox}[colback=yellow!5!white,colframe=gray!75!black,left=2pt,            
    right=2pt,           
    top=2pt,             
    bottom=2pt           
    ]
  \begin{tabular*}{\textwidth}{@{\extracolsep{\fill}}lr} #1   \\ \end{tabular*}
  {\bf{Input:}} #2  \\
  {\bf{Parameter:}} #3  \\  
  {\bf{Question:}} #4
  \end{tcolorbox}
}

\newcommand{\W}{\textsf{W}}

\newcommand{\D}{\ensuremath{\mathscr{D}}}
\newcommand{\T}{\ensuremath{\mathbb{T}}}

\newcommand{\family}[1]{{$(#1_1,\ldots,#1_{\ell})$}}
\newcommand{\callfamily}[1]{{$(#1_1,\ldots,#1_{\ell},\U,\F)$}}
\newcommand{\fami}[1]{{(#1_1,\ldots,#1_{\ell})}}

\newcommand{\solve}{\textsc{Solve}}
\newcommand{\intmath}{\operatorname{int}}

\newcommand\stree{\textsc{Steiner Tree}\xspace}
\newcommand{\Csc}{\textsc{Conn-SC}\xspace}
\newcommand\Sc{\textsc{Set cover}\xspace}

\newcommand\dslfree{$d$-semi-ladder-free\xspace}
\newcommand\deslfree{$(d+2)$-semi-ladder-free\xspace}

\newcommand\cds{\textsc{Conn-DS}\xspace}

\newcommand\ds{\textsc{Dominating Set}\xspace}

\newcommand{\gst}{\textsc{Group Steiner Tree}\xspace}

\newcommand\U{\mathcal{U}}
\newcommand\F {\mathcal{F}}

\newcommand{\semiladderfree}{semi-ladder-free\xspace}

\newcommand{\contains}{encapsulates\xspace}

\newcommand{\containing}{encapsulating\xspace}
\newcommand{\contain}{encapsulate\xspace}

\newcommand{\C}{\mathscr{C}}
\newcommand{\G}{\mathscr{G}}
\newcommand{\hG}{{\hat{G}}}
\newcommand{\hg}{{\hat{g}}}

\newcommand{\blocks}{demi-cores}

\newcommand{\NP}{\mathsf{NP}}
\newcommand{\coNP}{\mathsf{coNP}}
\newcommand{\poly}{\mathrm{poly}}

\begin{document}

\maketitle
\begin{abstract}
We study \textsc{Connected Dominating Set} on graphs whose closed-neighborhood set systems are $d$-semi-ladder-free. This structural condition strictly generalizes the biclique-free setting and provides a natural regime for connectivity-constrained domination. We obtain both a fixed-parameter algorithm and an approximate kernelization framework for the problem on this class.

Our algorithmic result is based on a new compact representation theorem for inclusion-wise minimal set covers in $d$-semi-ladder-free set systems. Although the number of minimal set covers of size at most $k$ may be as large as $n^{\Omega(k)}$, we show that all such set covers can nevertheless be encoded by a family of at most $k^{kd+1}$ tuples, and that this family can be enumerated in time $\Oh(k^{kd+2}\cdot nm)$. Combining this representation with a \textsc{Group Steiner Tree} subroutine, we obtain an algorithm for \textsc{Connected Set Cover}, which in turn yields an algorithm for \textsc{Connected Dominating Set} running in time $k^{kd+2}\cdot 2^k \cdot n^{\Oh(1)}$ and polynomial space.

For the preprocessing result, we introduce grouped domination cores and dominator cores, and prove polynomial upper bounds on their sizes in $d$-semi-ladder-free graphs. Using these structures, we obtain, for every fixed $d$ and $\varepsilon>0$, a polynomial-time $(1+\varepsilon)$-lossy compression for \textsc{Connected Dominating Set} to an equivalent reduced instance of size $k^{\Oh(d^2/\varepsilon)}$. The reduced instance is a \textsc{Connected Dominating Set} instance on a $(d+2)$-semi-ladder-free graph.


\end{abstract}

\newpage

\section{Introduction}
\textsc{Dominating Set} is a central $\NP$-hard problem in graph algorithms, and its connected
variant, \textsc{Connected Dominating Set} (\cds), arises in numerous applications
\cite{DBLP:conf/esa/GuhaK96, DBLP:conf/dialm/WuL99, DBLP:journals/monet/WanAF04, DBLP:conf/esa/NiklanovitsSVZ25}.
In parameterized complexity, \textsc{Dominating Set} and many of its variants are
$\W[2]$-hard when parameterized by the solution size $k$. Nevertheless, domination problems have
played a central role in the development of algorithmic techniques for parameterized algorithms
and polynomial-time preprocessing on structurally restricted graph classes. Our goal in this paper is to study \cds on broader structurally restricted graph classes than those considered so far.

From the preprocessing viewpoint, however, \cds is substantially harder than its unconstrained
counterpart: in particular, it does not admit a polynomial kernel even on sparse graph classes
such as $2$-degenerate or bounded-expansion graphs, unless $\NP \subseteq \coNP/\poly$
\cite{DBLP:journals/siamdm/EibenKMPS19}. This makes \emph{lossy kernelization}, or approximate
kernelization, a natural framework for the problem. In this setting, one allows a controlled
approximation factor $\alpha > 1$ in exchange for polynomial-size compression
\cite{DBLP:conf/stoc/LokshtanovPRS17}. Recent work has shown that this relaxation can be powerful:
for every $\alpha > 1$, \cds admits an $\alpha$-approximate polynomial kernel on sparse graph
classes such as biclique-free and bounded-expansion graphs \cite{DBLP:journals/siamdm/EibenKMPS19}.
More broadly, lossy kernelization has also proved effective in implicit covering settings, where
even implicitly defined \textsc{Hitting Set} instances admit efficient approximate compressions
\cite{DBLP:conf/esa/FominLL0TZ23}.

Motivated by this perspective, we study covering and domination problems under structural
restrictions, focusing on \textsc{Set Cover} and \textsc{Dominating Set} on
\emph{semi-ladder-free} set systems and graphs. This notion was introduced by
\cite{DBLP:conf/stacs/FabianskiPST19}.
\Ma{Our notion of semi-ladder-free graphs is equivalent to that in \cite{DBLP:conf/stacs/FabianskiPST19}.
For a \textsc{Set Cover} instance, $k$ is the size of the solution, $n$ denotes the size of the universe, and $m$ denotes the number of sets.}
Informally, a set system is $d$-semi-ladder-free if its
incidence structure avoids a certain combinatorial pattern of size $d$. This captures a broad
range of sparse settings, including nowhere-dense graph classes when neighborhoods are viewed as
set systems. Building on this framework, Guillemot~\cite{DBLP:journals/toct/Guillemot25} showed
that \textsc{Set Cover} is fixed-parameter tractable on semi-ladder-free hypergraphs, with running
time $k^{\Oh(dk)} \cdot n^{\Oh(1)}$, and additionally admits polynomial kernelization. These
results suggest that semi-ladder-freeness provides strong structural leverage for covering
problems. Our goal is to understand whether this leverage extends to
\emph{connectivity-constrained} variants.

A recurring idea in parameterized algorithms is that large families of minimal solutions may admit
succinct descriptions. The idea of succinctly encoding large families of minimal solutions has appeared in several forms in parameterized algorithms.
Guo et al.~\cite{DBLP:journals/siamdm/GuoNS11} introduced compact representations in the context of structured enumeration of minimal solutions, and Misra et al.~\cite{DBLP:journals/jco/MisraPRSS12} later obtained improved bounds in the setting of connectivity-constrained problems.
Related approaches based on representative sets show that exponentially many solutions can sometimes be compressed into polynomial-sized structures~\cite{DBLP:conf/soda/FominLS14}.
We build on these ideas in the setting of semi-ladder-free set systems.


\begin{tcolorbox}[colback=red!2,colframe=red!35,boxrule=0.5pt,arc=2mm,left=1mm,right=1mm,top=0.8mm,bottom=0.8mm]
\textbf{Our contribution.} 
The preceding developments suggest a natural question: can one combine the
structural power of semi-ladder-freeness with the algorithmic ideas of compact
representations and lossy kernelization to obtain efficient algorithms for
connectivity-constrained problems? In this work, we answer this question in the
affirmative.
\end{tcolorbox}


For \textsc{Connected Dominating Set}, FPT algorithms are known for several sparse graph classes. In particular, Telle and Villanger~\cite{TelleV19} showed that \textsc{Connected Dominating Set} is FPT on \(K_{t,t}\)-free graphs, while Eiben et al.~\cite{DBLP:journals/siamdm/EibenKMPS19} studied lossy kernelization for \textsc{Connected Dominating Set} on nowhere dense and bounded expansion classes. The enumeration of minimal dominating sets has also been studied extensively, with output-polynomial algorithms known for bounded-treewidth, bounded-clique-width, interval, chordal, planar, degenerate, split, permutation, line, and chordal bipartite graphs~\cite{CouturierHKKK13,KanteLMN14,GolovachHKKV16}. Our approach differs in that we do not enumerate all minimal dominating sets; instead, we exploit \(d\)-semi-ladder-freeness to obtain a compact representation containing only a bounded-size family of representatives sufficient for the subsequent connected-domination algorithm.

Our starting point is a structural theorem for \textsc{Set Cover}. Recall that a
\textsc{Set Cover} instance consists of a universe \(U\) and a family \(\mathcal{F}\subseteq 2^U\),
and the goal is to find a subfamily whose union is \(U\). Such a set cover is
\emph{inclusion-wise minimal} if none of its proper subfamilies is again a set cover. A direct
enumeration approach is hopeless in general: even on highly restricted instances, such as planar
instances with $K_{2,2}$-free incidence graphs, the number of minimal set covers of size at most
$k$ can be as large as $m^k$. We show, however, that this apparent combinatorial explosion can be
circumvented. More precisely, we organize all such solutions into a bounded family of tuples of
pairwise disjoint subfamilies, where each tuple compactly encodes many minimal set covers: one
obtains a solution by selecting exactly one set from each subfamily
(\Cref{def:compact_representations}). 

For $d$-semi-ladder-free set systems, this yields a compact
representation of size $k^{\Oh(kd)}$, and such a representation can be enumerated within the same
bound. 


\begin{restatable}[\app]{theorem}{findcompactrepresentations}
    \label{theorem:find_compact_representations}
    Let $(\cU,\cF)$ be a \dslfree set system, and let $k \in \mathbb{Z}_{>0}$. 
    Then, there exists a family $\mathscr{D}$ of tuples, each of arity at most $k$,
    that compactly represents all inclusion-wise minimal set covers of size at most $k$ in $(\cU,\cF)$ such that $
        |\mathscr{D}| \le k^{kd+1}$.
    Moreover, the family $\mathscr{D}$ can be enumerated in $k^{kd+2} \cdot \Oh(nm)$ time, where $n \coloneq |\cU| $ and $ m \coloneq |\cF|$.
\end{restatable}

We also prove lower bounds showing
that this bound is tight up to polynomial factors in the exponent, and additionally establish an unconditional
impossibility result ruling out \fpt-sized compact representations for general set systems
(\Cref{subsubsec:tighteness and impossibility,thm:tightness-of-compact-rep-in-general-set-systems}). 

This structural result directly leads to algorithms for connectivity-constrained covering problems.
Using the compact representation together with a \textsc{Group Steiner Tree} subroutine, we obtain
a fixed-parameter algorithm for \Csc on $d$-semi-ladder-free set
systems with running time $k^{\Oh(kd)} \cdot 2^k \cdot n^{\Oh(1)}$ and polynomial space.
\begin{restatable}[\app]{theorem}{cscalgo}
    \label{thm:connected-set-cover}
    There exists an algorithm for \Csc on \dslfree set systems that runs in $k^{kd+2}\cdot 2^k \cdot n^{\Oh(1)}$ time and uses polynomial space.
\end{restatable}
Consequently, we also obtain an analogous fixed-parameter
algorithm for \cds on $d$-semi-ladder-free graphs (\Cref{corrolary:conn-dominating-set}).

Our second contribution concerns preprocessing. To exploit the structure of semi-ladder-free
graphs, we use the notion of \emph{dominator cores}, which has been considered earlier in related
settings, and introduce a new notion of \emph{grouped domination cores}
(\Cref{def : Dominator Core,def : dominating set and families,def : Grouped Domination Core}).
These notions capture redundancy in domination constraints. Roughly speaking, a dominator core
partitions the relevant dominators into classes such that every inclusion-wise minimal dominating set of size at most $k$ contains at most one
vertex from each class, making the vertices of a class interchangeable. A grouped domination core
strengthens this idea by requiring each group to be dominated as a whole, thereby aggregating many
domination constraints into a single requirement. We show that every $d$-semi-ladder-free graph
admits a grouped domination core of size at most $k^d$ and a dominator core of size
$\Oh(k^{d^2})$, both computable in polynomial time. 

\begin{restatable}[\app]{theorem}{existencecores}
    \label{thm:existence-domination-dominator-core}
    Let $G$ be a $d$-semi-ladder-free graph and let $k\in\mathbb N$. One can compute, in polynomial time, a grouped domination core of size at most $k^d$ and a dominator core of size at most $\Oh(k^{d^2})$.
\end{restatable}

Finally, we return to lossy preprocessing. Since exact kernelization is unlikely for \cds already
on bounded-degeneracy graphs---a subclass of bounded \dslfree graphs---the right target here is an
approximate kernel \cite{DBLP:conf/esa/CyganGH13}. Building on the lossy-kernel framework of
\cite{DBLP:journals/siamdm/EibenKMPS19}, we prove that \cds on $d$-semi-ladder-free graphs admits
a polynomial-size approximate compression into instances on $(d+2)$-semi-ladder-free graphs. %

\begin{restatable}[\app]{theorem}{lossycompression} \label{thm: lossy_compression}
For every fixed $d$ and $\epsilon>0$, \cds on $d$-semi-ladder-free graphs admits a polynomial-time $(1+\epsilon)$-lossy compression of size $k^{\Oh(d^2/\epsilon)}$ to \cds on $(d+2)$-semi-ladder-free graphs.
\end{restatable}

This extends the reach of lossy preprocessing beyond the graph
classes considered previously. 
Proofs and additional details for statements marked $(\dagger)$ appear in the appendix.

\section{Preliminaries}
\label{sec:prelims}


We begin by fixing standard graph-theoretic notation that will be used throughout
the paper. These notions are well known, but we recall them here to keep the
presentation self-contained.

The vertex and edge sets of a graph $G$ are denoted by $V(G)$ and $E(G)$, respectively.
For a subset $S \subseteq V(G)$, $G[S]$ denotes the subgraph of $G$ induced by $S$.
If $G[S]$ is connected, then we say that the set $S$ is \emph{connected}.
We denote the \emph{open neighborhood} of a vertex $v$ in $G$ by $N_G(v)$,
and its \emph{closed neighborhood} by $N_G[v]$.
For a set $S \subseteq V(G)$, we denote by $N_G[S]$ the closed neighborhood of the vertices in $S$.
The degree of a vertex $v$ in the induced subgraph $G[S]$ is denoted by $d_{G[S]}(v)$,
and its degree in $G$ by $d_G(v)$.
We say that a vertex~$u$ \emph{dominates} a vertex~$v$ in~$G$ if~$(u,v)\in E(G)$ or~$u=v$.
With these basic notions in place, we can now formally state the central graph problem
studied in this work.

\defparprob{{\sc Connected Dominating Set} (\cds)}
{An undirected graph~$G$ and an integer~$k$}
{Solution size ($k$)}
{Find a set~$S$ of size at most~$k$ such that~${G}[S]$ is connected and~$N_G[S]=V({G})$.}

Our approach relies on a structural restriction that originates from set systems.
We therefore, use the notion of semi-ladder-freeness for set families as defined in \cite{DBLP:journals/toct/Guillemot25}.

\begin{definition}[\dslfree set family]\label{def:semiladder-set-family}
    We call a set family~$(\mathcal{U},\mathcal{F})$ \dslfree if there do not exist~$d$ sets~$S_1, S_2, \ldots, S_d \in \mathcal{F}$ and~$d$ elements~$x_1, x_2, \ldots, x_d \in \mathcal{U}$ such that for all~$i,j \in [d]$,
$x_i \in S_j$ if~$i < j$ and
$x_i \notin S_i$.
\end{definition}

To apply this concept to graphs, we interpret graphs through their closed neighborhoods,
which naturally give rise to a set family.

\begin{definition}[Semi-ladder-free graphs] \label{def:graph-semiladder}
Let~$G=(V,E)$ be an undirected graph.  Its closed-neighborhood set system is the indexed set system $(\mathcal U,\mathcal F_G)$ with $\mathcal U=V$ and one set $F_v=N_G[v]$ for every vertex $v\in V$.  We allow two different vertices to define the same subset of $V$; whenever we speak about connectivity or solutions, the index $v$ is retained.  We call $G$ \emph{$d$-semi-ladder-free} if the underlying family of closed neighborhoods is $d$-semi-ladder-free.

The \emph{semi-ladder index} of $G$ is the smallest integer $q$ such that $G$ is $(q+1)$-semi-ladder-free.  Throughout the paper, $d$ is treated as a fixed constant (or, equivalently, the algorithms are run with $d$ supplied as a parameter of the graph class).
\end{definition}

If a graph has a dominating set of size at most one, then this set is connected and can be found in polynomial time.  We therefore freely handle this case separately.  The main connectivity primitive used later is the following.
\ifthenelse{\equal{\version}{compressed}}
{
}
{
{\subparagraph*{The \stree Problem.}
In the \stree problem, we are given a graph~$G$, a subset~$R \subseteq V(G)$ of \emph{terminals},
and a weight function~$w : E(G) \rightarrow \mathbb{N}$.
A \emph{Steiner tree} is a subtree~$T$ of~$G$ such that~$R \subseteq V(T)$,
and its \emph{cost} is~$w(T) = \sum_{e \in E(T)} w(e)$.
The task is to find a minimum-cost Steiner tree.
We may assume without loss of generality that~$G$ is complete and that~$w$ satisfies the triangle inequality,
that is, for all~$u,v,w \in V(G)$ we have~$w(uw) \le w(uv) + w(vw)$.
This can be achieved by adding edges~$uv$ with weight equal to the shortest path distance between~$u$ and~$v$,
and keeping only the lightest edge when multiple edges arise.}
}

\subparagraph*{The \gst Problem.}
We use the following standard connectivity subroutine.  In the \gst problem, the input is an undirected graph $G=(V,E)$ and $\ell$ pairwise vertex-disjoint subsets $S_1,\ldots,S_\ell\subseteq V(G)$, called \emph{groups}.  The goal is to compute a minimum-size tree containing at least one vertex from every group.  The parameter is the number $\ell$ of groups.

\begin{proposition}[\cite{DBLP:journals/jco/MisraPRSS12}]
    \label{prop : group steiner tree algo}
    The \gst problem can be solved in time
    $2^{\ell} \cdot n^{\Oh(1)}$ using polynomial space.
\end{proposition}

Now, we are ready to argue that the class of $6d$-semi-ladder-free graphs is a superclass of $K_{d,d}$-free (i.e., biclique-free) graphs, which have been the focus of previous work on lossy kernels for \cds.
This establishes that our results strictly generalize previous work.
We introduce a semi-induced version of semi-ladders to establish this connection; edges inside the two sides are irrelevant for this definition.

\begin{definition} \label{def:induced-semiladder-free-graph} 
A graph~$G$ is \emph{semi-induced $d$-semi-ladder-free} if there do not exist two disjoint vertex sets $L = \{a_1, a_2, \ldots, a_d\}$ and $R = \{b_1, b_2, \ldots, b_d\}$, such that in the induced subgraph $G[L \cup R]$, for all $i,j \in [d]$ with $i < j$, we have $(a_i,b_j) \in E(G)$, and for all $i \in [d]$, we have $(a_i,b_i) \notin E(G)$.
\end{definition}


\begin{restatable}[\cite{fabiański2018progressivealgorithmsdominationindependence}]{lemma}{semiladderfreegraph}
    \label{lem:semiladder-free-graph}
    Every $K_{d,d}$-free graph is a $3d$-semi-ladder-free graph.
\end{restatable}
\ifthenelse{\equal{\version}{compressed}}
{
}
{
    \begin{proof}
        Suppose, for contradiction, that $G$ contains a semi-induced $2d$-semi-ladder.
        Then there exist disjoint vertex sets
        \[
            L = \{a_1, a_2, \ldots, a_{2d}\}
            \quad\text{and}\quad
            R = \{b_1, b_2, \ldots, b_{2d}\}
        \]
        such that $G[L \cup R]$ is bipartite with bipartition $(L,R)$ and satisfies the semi-ladder condition:
        for all $i,j \in [2d]$ with $i < j$, we have $(a_i,b_j) \in E(G)$, and for all $i \in [2d]$, we have $(a_i,b_i) \notin E(G)$.
       Consider the subsets $
            L' = \{a_1, a_2, \ldots, a_d\}
            \quad\text{and}\quad
            R' = \{b_{d+1}, b_{d+2}, \ldots, b_{2d}\}$.
        Let $a_i \in L'$ and $b_j \in R'$.
        Then $i \le d$ and $j \ge d+1$, hence $i < j$.
        By the semi-ladder condition, $(a_i,b_j) \in E(G)$.
        Therefore, every vertex of $L'$ is adjacent to every vertex of $R'$, and $G[L' \cup R']$ contains a copy of $K_{d,d}$.
        This contradicts the assumption that $G$ is $K_{d,d}$-free.
        Hence, $G$ cannot contain a semi-induced $2d$-semi-ladder.
    \end{proof}
}

Now we are ready to establish the relation between semi-induced semi-ladder-free graphs and semi-ladder-free graphs.

\begin{restatable}[\app]{lemma}{semiladderfreegraphtwo}
\label{lem:semiladder-free-graph-2}
Every semi-induced $d$-semi-ladder-free graph is $3d$-semi-ladder-free.
\end{restatable}
\ifthenelse{\equal{\version}{compressed}}
{
}
{
    \begin{proof}
    Suppose, for contradiction, that $G$ contains a $3d$-semi-ladder.
    Then there exist vertices and neighborhoods
    \[
    a_1,\dots,a_{3d}
    \quad\text{and}\quad
    N[b_1],\dots,N[b_{3d}]
    \]
    such that
    \[
    a_i \notin N[b_i]
    \quad\text{and}\quad
    a_i \in N[b_j] \text{ for all } i<j.
    \]
    
    Note that if the vertices $a_1,\dots,a_{3d}$ and $b_1,\dots,b_{3d}$ were all distinct, we would have found a $3d$-semi-ladder as an induced subgraph in $G$.
    This would contradict the assumption that $G$ is semi-induced $d$-semi-ladder-free, since the sets $N[b_i]$ can be correspond to $b_i$ themselves while maintaining the neighborhood relation.
    
    Now, consider the following greedy process. We start with $i=1$. By definition of a semi-ladder, $a_i \notin N[b_i]$ but is adjacent to $N[b_j]$ for all $j>i$. We add $a_i$ and $b_i$ to the sets $L$ and $R$, respectively. Note that $a_i \neq b_i$ as $a_i \notin N[b_i]$. We then follow a row marking process. Given $a_i$ and $b_i$, we mark the following rows:
    \begin{itemize}
        \item Row $j$ where $a_j=b_i$,
        \item Row $j$ where $b_j=a_i$.
    \end{itemize}
    Hence, each time we add a vertex to $L$ and $R$, we mark at most two rows. We find a row $i'$ that is not marked and add $a_{i'}$ and $b_{i'}$ to $L$ and $R$, respectively. We repeat this process until we have added $d$ vertices to $L$ and $R$. We can always find such an unmarked row as we start with $3d$ rows and each time we add a vertex to $L$ and $R$, we mark at most two rows.
    
    After this process, we have two sets $L$ and $R$ of size $d$ each such that for all $i,j \in [d]$ with $i < j$, we have $(a_i,b_j) \in E(G)$, and for all $i \in [d]$, we have $(a_i,b_i) \notin E(G)$. This contradicts the assumption that $G$ is semi-induced $d$-semi-ladder-free.
    \end{proof}
}

\begin{definition}
\label{def:approxreduction}
Let~$\alpha \ge 1$ and let~$\Pi$ and~$\Pi'$ be two parameterized optimization problems.  
An \textbf{$\alpha$-approximate polynomial parameter transformation} (\emph{$\alpha$-appt})  
${\cal A}$ from~$\Pi$ to~$\Pi'$ consists of 
a polynomial-time \emph{reduction algorithm}~${\cal R}_{\cal A}$ and  
a \emph{solution lifting algorithm}.

Given an instance~$(I,k)$ of~$\Pi$, the reduction algorithm outputs an instance~$(I',k')$ of~$\Pi'$.  
Given a solution~$s'$ to~$(I',k')$, the lifting algorithm produces a solution~$s$ to~$(I,k)$.  
If~$\Pi$ is a minimization problem, then
\[
\frac{\Pi(I,k,s)}{OPT_{\Pi}(I,k)} \le \alpha \cdot \frac{\Pi'(I',k',s')}{OPT_{\Pi'}(I',k')},
\]
and for maximization problems, the inequality is reversed.  
We denote this by~$\Pi \prec_{\alpha\text{-appt}} \Pi'$.
\end{definition}

\begin{definition}
\label{def:approxcompression}
Let~$\alpha \ge 1$.  
An \textbf{$\alpha$-approximate compression} from~$\Pi$ to~$\Pi'$  
is an $\alpha$-appt~${\cal A}$ such that the size of the reduced instance,  
$\text{size}_{\cal A}(k) = \sup\{|I'|+k' : (I',k') = {\cal R}_{\cal A}(I,k),~I \in \Sigma^*\}$,
is upper bounded by a computable function~$g : \mathbb{N} \to \mathbb{N}$.

Let $\alpha \ge 1$ and let $\Pi$ be a parameterized optimization problem.  
An \textbf{$\alpha$-approximate compression} of $\Pi$ into itself (or another problem $\Pi'$) is also called a \emph{lossy compression}.  
Formally, it is an $\alpha$-approximate compression (as in \Cref{def:approxcompression}) that maps an instance $(I,k)$ of $\Pi$ to an instance $(I',k')$ of $\Pi'$ of size bounded by a function $g(k)$, such that solutions to $(I',k')$ can be lifted to solutions of $(I,k)$ with approximation factor $\alpha$.
We refer to the reduced instance $(I',k')$ as the \emph{compression} of $(I,k)$.
\end{definition}

\newcommand{\cS}{\mathcal{S}}

\section{Connected Set Cover of $d$-semi-ladder-free set systems}
\label{sec:csc-dslfree}
In this section, we first define the \textsc{Connected Set Cover} problem (\Csc) and then present a
fixed-parameter tractable (FPT) algorithm for \Csc on \dslfree set systems.
We subsequently use this result to obtain an FPT algorithm for \cds on
\dslfree graphs. 


\defparprob
{{\sc Connected Set Cover} (\Csc) }
{A graph $\hat{G}$ defined over a set system $(\mathcal{U}, \mathcal{F})$ with
    $V(\hat{G}) = \mathcal{F}$, and an integer $k$}
{Solution size, $k$}
{Find a subfamily $S \subseteq \mathcal{F}$ of size at most $k$ such that the
    graph induced on the vertices corresponding to the sets in $S$ is connected
    and $S$ covers $\mathcal{U}$.}

\hide{\ma{Move the boxed def here?}
Given a set family $(\mathcal{U}, \mathcal{F})$ defined over a graph $\hat{G}$,where $V(\hat{G}) = \mathcal{F}$, the \Csc problem asks for a subfamily$S \subseteq \mathcal{F}$ such that $S$ covers $\mathcal{U}$ and the graph induced by the vertices corresponding to the sets in $S$ is connected.}

Intuitively, \Csc can be seen as a connectivity-constrained variant of
    {\sc Set Cover}, where solutions are required to form a connected subgraph
in the underlying graph $\hat{G}$.

Our motivation for studying \Csc stems from its close relationship with
\cds.
It is straightforward to see that any instance of \cds can be modeled as an instance
of \Csc by considering the closed neighborhoods of the vertices as sets in $\mathcal{F}$. We will formally prove this connection later in this section. Let us proceed to describe our algorithm for \Csc on \dslfree set systems.

In particular, we will exploit the fact that ideas developed for
\Sc on \dslfree set systems can be lifted to the connected
setting. One possible approach would be to first compute a set cover
using the algorithm of~\cite{DBLP:journals/toct/Guillemot25} for \Sc on
\dslfree set systems, and then connect the selected sets using an FPT algorithm for
    {\sc Steiner Tree}.
    For the special case of \cds, this yields a simple $2$-approximation \fpt algorithm.

Instead of following this two-step approach, we take a more direct route.
We generalize the ideas of~\cite{DBLP:journals/toct/Guillemot25} and show that
the structural properties underlying their results can be leveraged to
compute a \emph{compact representation} of all inclusion-wise minimal set
covers of size at most $k$ in time $k^{\Oh(kd)} \cdot n^{\Oh(1)}$.
Using this compact representation, we then design an \fpt algorithm for
\Csc on \dslfree set systems.
We next discuss the high-level idea underlying our algorithm for \Csc on
\dslfree set systems.

\subparagraph*{Overview of the idea.} The main idea of our algorithm is to first introduce a set of tuples of disjoint subfamilies of $\F$ such that each minimal set cover of size at most $k$ can be obtained by picking one set from each subfamily of a tuple. We call such a set of tuples a \emph{compact representation} of all minimal set covers. 

Once we establish that each minimal set cover of size at most $k$ can be obtained from a tuple, we can use the \gst algorithm given by \Cref{prop : group steiner tree algo} to find a connected set cover. The harder part is to ensure that picking one set from each subfamily of a tuple \emph{always} gives a valid set cover; otherwise, we cannot use \Cref{prop : group steiner tree algo} directly.

We formalize the above idea in the next subsection by defining the notion of a compact representation of all minimal set covers and presenting an algorithm to compute such a representation for \dslfree set systems.


\subsection{Compact Representation of all minimal set covers of a \dslfree set system}

\label{sec:compact-representation-dslfree-system}
In many covering and hitting problems, inclusion-wise minimal solutions provide a useful starting point. If one can
(i) prove an \fpt upper bound on the number of inclusion-wise minimal solutions (with respect to the relevant
parameter---such as the solution size), and (ii) enumerate all such minimal solutions in \fpt time, then a wide range
of constrained variants becomes approachable via a common approach.

The key observation is that in several constrained variants, any feasible solution {\em contains}  some inclusion-wise minimal solution for the underlying unconstrained problem. Once an enumeration procedure is available, we may iterate over all minimal
solutions $\mathcal{C}$, and for each $\mathcal{C}$ attempt to \emph{extend} it to a solution satisfying the
additional requirements, using problem-specific algorithmics. For instance, if the added requirement is
connectivity, one may start from $\mathcal{C}$ and then augment it by extra objects to enforce connectivity, in the
spirit of a Steiner-type completion.

However, such an \fpt bound need not exist in general. In the next subsection we show, in the context of {\sc Set Cover},
that the number of inclusion-wise minimal set covers can be as large as $m^{k}$ even on restricted families of
instances.

\subsubsection{There can be XP-many minimal set covers}
\label{subsec:minsclarge}
To this end, we first recall the definition of the {\sc Set Cover} problem.
An instance of {\sc Set Cover} is given by a set system $(\mathcal{U},\mathcal{F})$,
where $\mathcal{U}$ is a finite universe and $\mathcal{F} \subseteq 2^{\mathcal{U}}$
is a family of subsets.
A subfamily $\mathcal{C} \subseteq \mathcal{F}$ is called a \emph{set cover} if
$\bigcup_{S \in \mathcal{C}} S = \mathcal{U}$.
Such a set cover $\mathcal{C}$ is said to be \emph{inclusion-wise minimal} if no
proper subfamily of $\mathcal{C}$ is also a set cover.
In what follows, we give a family of instances that admit $\Omega\left(\left({\frac{|\F|}{k}}\right)^{k}\right)$ inclusion-wise minimal set covers,
even under strong restrictions on the structure of the instance---in particular, even when the bipartite
incidence graph of $(\mathcal{U},\mathcal{F})$ is planar and $K_{2,2}$-free.

\subparagraph{Construction.}
Let $k,m\in\mathbb{N}$. Define the universe
$
\mathcal U \;:=\; \{t\} \cup \{u_i \mid i\in [k]\} \cup \{b_{i,j}\mid i\in[k],\, j\in[m]\}.
$
Let
$
B := \{t\}\cup\{b_{i,j}\mid i\in[k],\, j\in[m]\},
$
and for each $i\in[k], j\in[m]$,
let
$
S_{i,j} := \{u_i,\, b_{i,j}\}
.
$
Define the family $\mathcal F:=\{B\}\cup\{S_{i,j}:i\in[k],j\in[m]\}$.
All sets in $\mathcal F$ are pairwise distinct (each $S_{i,j}$ contains a unique element $b_{i,j}$, and $B$ contains $t$).
An illustration of the incidence graph of $(\mathcal U,\mathcal F)$ is given in Figure~\ref{fig:planar-incidence}.


\begin{restatable}[\app]{theorem}{planarmanyminimal}
\label{thm:planar-many-minimal}
For every $k,m\in\mathbb{N}$ there exists a set system $(\mathcal U,\mathcal F)$ such that:
(i) all sets in $\mathcal F$ are distinct,
(ii) $(\mathcal U,\mathcal F)$ has $\Omega\left(\left({\frac{|\F|}{k}}\right)^{k}\right)$ inclusion-wise minimal set covers, and
(iii) the incidence graph of $(\mathcal U,\mathcal F)$ is planar and $K_{2,2}$-free.
\end{restatable}

\ifthenelse{\equal{\version}{compressed}}
{
}
{
\begin{proof}
Applying the following two results together, on the construction, gives us the desired result.
\begin{lemma}\label{lem:many-minimal}
The instance $(\mathcal U,\mathcal F)$ has exactly $m^{k}$ inclusion-wise minimal set covers.
\end{lemma}
\begin{proof}
Any set cover $\mathcal C\subseteq\mathcal F$ must contain $B$, since $t$ appears only in $B$.
Once $B$ is chosen, every $b_{i,j}$ is covered, so the only remaining elements are $u_1,\dots,u_k$.
For each $i\in[k]$, the element $u_i$ appears only in $S_{i,1},\dots,S_{i,m}$, hence $\mathcal C$ must contain at least
one set from this block.

If $\mathcal C$ is inclusion-wise minimal, then for each $i$ it contains \emph{exactly one} set from
$\{S_{i,1},\dots,S_{i,m}\}$: if it contained two, one would be redundant because $B$ already covers all $b_{i,j}$'s.
Thus every inclusion-wise minimal set cover has the form
$\mathcal C \;=\; \{B\}\ \cup\ \{S_{1,j_1},S_{2,j_2},\dots,S_{k,j_k}\}$ with $j_i\in[m]\ \text{for all } i\in[k]$. 
Conversely, any such choice covers $\mathcal U$ and is inclusion-wise minimal (removing $B$ uncovers $t$, and removing
$S_{i,j_i}$ uncovers $u_i$). Hence there are exactly $m^{k}$ such covers.
\end{proof}

\begin{lemma}\label{lem:planar-k22}
The incidence graph of $(\mathcal U,\mathcal F)$ is planar and $K_{2,2}$-free (and hence $K_{d,d}$-free for every $d\ge 2$).
\end{lemma}
\begin{proof}
For planarity, place the set-vertex $B$ in the center and the element-vertices $u_1,\dots,u_k$ on a circle around $B$.
For each $i\in[k]$, reserve a disjoint wedge region between $u_i$ and $B$, and inside it draw $m$ internally disjoint
paths $u_i-S_{i,j}-b_{i,j}-B$ (for $j\in[m]$). Place $t$ near $B$ and draw the edge $tB$. No edges cross, so the
incidence graph is planar. See Figure~\ref{fig:planar-incidence} for an illustration. 

For $K_{2,2}$-freeness, it suffices to show that any two distinct sets in $\mathcal F$ intersect in at most one element.
Indeed, for $S_{i,j}$ and $S_{i',j'}$, the intersection is empty if $i\neq i'$ and equals $\{u_i\}$ if $i=i'$ and
$j\neq j'$. Moreover, $B\cap S_{i,j}=\{b_{i,j}\}$. Hence no two sets share two elements, so the incidence graph has no
$K_{2,2}$.
\end{proof}
\end{proof}
}


\subsubsection{Compact Representation via Tuples}
In this section, we introduce a framework to represent all inclusion-wise minimal
set covers of a given instance of {\sc Set Cover} in a compact way.
The goal is to construct a small collection of tuples, each of which implicitly
encodes many minimal set covers.

We now introduce the central combinatorial object that will allow us to group many
inclusion-wise minimal set covers together. For the time being, we relax the condition that the set families need to be pairwise disjoint. We will see that this condition will be implicitly satisfied when we enforce the minimality condition on the set covers represented by the tuples, while ensuring that if we pick one set from each subfamily of a tuple, we always get a valid set cover. Here is the formal definition of an $\ell$-tuple.


For a set family $\Co{S} \sse \Co{F}$, a permutation of the elements of $\Co{S}$, denoted by $(S_1, \ldots, S_{|\Co{S}|})$, is called a \emph{sequence} of $\Co{S}$.
Moreover, it is to be understood that any $\ell$-sized tuple of sets $(S_1, \ldots, S_{\ell})$ is a sequence of the set family $\{S_1, \ldots ,S_{\ell}\}$. 

\begin{definition}[A tuple encapsulating a sequence]
    \label{def:tuple}
    Given a set system $(\cU,\cF)$ and an integer $\ell\in\mathbb{Z}_{>0}$, consider 
    an $\ell$-tuple $\T$\hide{ is a sequence of $\ell$ } of subfamilies of $\cF$. That is, $
        \T=(F_1,F_2,\ldots,F_\ell)$, where  $F_i\subseteq \cF$, for each $i\in[\ell]$.

    We say that an $\ell$-tuple $\T $ \emph{\contains} a sequence of sets $(S_1,\ldots,S_{\ell})$ if $(S_1,\ldots,S_{\ell}) \in F_1 \times \ldots \times F_{\ell}$. 
    The subfamilies in an $\ell$-tuple need not be pairwise disjoint or distinct\footnote{To avoid confusion, we will exclusively refer to a tuple of sets, i.e., $(S_1, \ldots, S_\ell)$ where each $S_i$ is a set, as a ``sequence''. \hide{; and ``tuple'' to refer to $(F_1, \ldots, F_\ell)$, where each $F_i$ is a set-family}}.
\end{definition}

From Section~\ref{subsec:minsclarge}, we know that the number of minimal set covers of size $k$ can be potentially very large. Each tuple allows us to group many such minimal set covers together. However, one group might not be sufficient to capture all minimal set covers of size at most $k$. Therefore, we will need a collection of tuples which will together capture all minimal set covers of size at most $k$. We call such a collection \emph{a compact representation} of all minimal set covers. Next, we formally define this notion.

\begin{definition}[Compact representations of all minimal set covers]
    \label{def:compact_representations}
    Let \(\mathcal{I}=(\mathcal{U},\mathcal{F},k)\) be an instance of {\sc Set Cover}.
    A family \(\mathscr{D} \sse \bigcup_{\ell \in [k]}(2^{\mathcal{F}})^{\ell}\) of tuples of subfamilies of \(\mathcal{F}\) is said to \emph{compactly represent}\hide{ all} inclusion-wise minimal set covers of size at most \(k\) if  the following hold:

    \begin{enumerate}
        \item \label{compact_rep:covering_condition} \emph{(Covering condition)} For each tuple $\mathbb{T}=(F_1,\ldots,F_{\ell}) \in \mathscr{D}$ where $\ell \le k$ and for each tuple $(S_1,\ldots, S_\ell) \in F_1\times \ldots \times F_{\ell}$, we have $\bigcup_{i=1}^\ell S_i = \mathcal{U}$.
              
        \item \emph{(Disjointness condition)} For each $\mathbb{T}=(F_1,\dots,F_{\ell}) \in \mathscr{D}$ where $\ell \le k$, the families $\{F_i: i \in [\ell]\}$ are pairwise-disjoint, i.e., no two families have a set in common. 
        \item \emph{(Completeness condition)} For each inclusion-wise minimal set cover \(\mathcal{S}\subseteq \mathcal{F}\) with \(|\mathcal{S}|\le k\), there exists a tuple  $\mathbb{T}=(F_1,\dots,F_{|\Co{S}|}) \in \mathscr{D}$ such that for a sequence $(S_{1}, \ldots, S_{|\Co{S}|})$ of $\Co{S}$, we have \(S_i\in F_i\), for each \(i\in[|\mathcal S|]\).
        
    \end{enumerate}
    The \emph{size} of the compact representation is \(|\mathscr{D}|\).
\end{definition}

Taken together, these conditions ensure that every tuple represents only valid set
covers, that different positions in a tuple are structurally separated, and that
every inclusion-wise minimal set cover of size at most $k$ is captured by at least
one tuple in the family. \Ma{We note here that we cannot compute the number of all inclusion-wise minimal set covers using compact representations by simply multiplying the numbers and summing them up as a minimal dominating set might be present in multiple tuples.}

Next, we will show how the above definition helps us to compactly represent the set covers in the instance constructed in \Cref{subsec:minsclarge}.

\begin{remark}[An Instance with a Single-Tuple Compact Representation.]
    \label{rem:compact-size-1}
    We note that a compact representation can be much smaller than the number of minimal set covers it represents.
    For example, consider the set system $(\cU,\cF)$ from the construction in \Cref{subsec:minsclarge}, where $\cF=\{B\}\ \cup\ \{S_{i,j}\mid i\in[k],\, j\in[m]\}$.
    There are $m^k$ many inclusion-wise minimal set covers of size at most $k+1$.
    However, all of them can be compactly represented by a single $(k+1)$-tuple $\T^{\star} \;:=\; (F_0,F_1,\ldots,F_k)$, where $F_0:=\{B\}$ and $
    F_i:=\{S_{i,1},S_{i,2},\ldots,S_{i,m}\}$ for each $i\in[k]$.
    The proof of its validity is deferred to \Cref{sec:app-csc-dslfree}.
\end{remark}




\ifthenelse{\equal{\version}{compressed}}
{
}
{
\begin{lemma}\label{lem:compact-size-1}
The family $\mathscr{D}=\{\T^{\star}\}$ is a compact representation (in the sense of Definition~\ref{def:compact_representations})
of all inclusion-wise minimal set covers of size at most $k+1$ for the instance $(\cU,\cF,k+1)$.
In particular, $|\mathscr{D}|=1$.
\end{lemma}
\begin{proof}
We verify the three conditions in Definition~\ref{def:compact_representations} for the single tuple
$\T^{\star}=(F_0,F_1,\ldots,F_k)$.

\smallskip
\noindent\emph{Covering condition.}
Fix any choice $(f_0,f_1,\ldots,f_k)\in F_0\times F_1\times\cdots\times F_k$.
Then $f_0=B$. Moreover, for each $i\in[k]$, we have $f_i=S_{i,j_i}$ for some $j_i\in[m]$, hence $u_i\in f_i$.
Therefore $B$ covers $t$ and all elements $b_{i,j}$, while $\{f_i:i\in[k]\}$ covers all $u_1,\ldots,u_k$.
Thus $\bigcup_{i=0}^k f_i=\cU$.

\smallskip
\noindent\emph{Structural condition.}
The subfamilies $F_0,F_1,\ldots,F_k$ are pairwise disjoint as families of sets: $F_0$ contains only $B$, and for
each $i\in[k]$ the family $F_i$ consists only of sets $S_{i,j}$ with fixed first index $i$.

\smallskip
\noindent\emph{Completeness condition.}
Let $\cS\subseteq \cF$ be an inclusion-wise minimal set cover with $|\cS|\le k+1$.
By Lemma~\ref{lem:many-minimal} (or by the same argument), every inclusion-wise minimal set cover in this instance has
the form $\cS=\{B\}\cup\{S_{1,j_1},\ldots,S_{k,j_k}\}$,
for some $(j_1,\ldots,j_k)\in[m]^k$, and hence $|\cS|=k+1$.
Order $\cS$ as $(B,S_{1,j_1},\ldots,S_{k,j_k})$. Then $B\in F_0$ and $S_{i,j_i}\in F_i$ for each $i\in[k]$,
so $\cS$ is captured by $\T^{\star}$.

This establishes all three conditions. Hence $\mathscr{D}$ is a compact representation and has size $1$.
\end{proof}
}

While the subfamilies comprising a tuple might contain many sets, they can be processed efficiently if they possess a compact implicit representation. 
Specifically, our algorithm represents each subfamily $F$ via a base subset of elements $A \subseteq \mathcal{U}$, such that every set in $F$ is a superset of $A$. 
Consequently, we define the notion of an extension of a sequence of sets.


\begin{definition}[Extension of a sequence of sets]
    Let $(\U,\F)$ be a set family.
    Let $(A_1,\ldots,A_\ell)$ be a sequence, where $A_i\subseteq \U$, for each $i \in [\ell]$. 
    We say that a sequence $(S_1,\ldots,S_\ell)$ \emph{extends} $(A_1,\ldots,A_\ell)$ if for each $i\in [\ell]$, we have $A_i \subseteq S_i \in \F$. 
\end{definition} 
Meanwhile, we will also ensure that the union of the small representations of the subfamilies (the sets $A_i$) in a tuple always covers the entire universe.
Consequently, any extension of a \emph{latent cover}, defined below, will also cover the universe and hence form a valid set cover. We formalize this idea in the next definition.

\begin{definition}[Latent Cover]
    We say that a sequence of sets
    $(A_1,\ldots,A_{\ell})$ forms a \emph{latent cover} of a set cover instance $({\cal U}, {\cal F} , k)$ where $\ell \le k$, if $\bigcup^{\ell}_{i=1} A_i = {\cal U}$. Note that the sets $A_i$ for each $i \in [\ell]$ need not be present in the set family $\cal F$.
\end{definition}



This leads us to the following observation.
\begin{observation}
    For a set cover instance $({\cal U}, {\cal F} , k)$, let $(A_1,\ldots,A_{\ell})$ be a latent cover where $\ell \le k$.
    An\hide{minimal} extension
    \((S_1,\ldots,S_\ell)\) of a latent cover \((A_1,\ldots,A_\ell)\) is a set cover.
    \end{observation}

We show that for \dslfree set systems, all inclusion-wise minimal set covers of size at most $k$ admit a compact representation that can be efficiently enumerated, as formalized in \Cref{theorem:find_compact_representations}. \Ma{The algorithm in \Cref{alg:find_compact_representations} is motivated by \cite{DBLP:journals/toct/Guillemot25}}.
The following notion of intersection closure of a set family, due to Guillemot~\cite{DBLP:journals/toct/Guillemot25}, will be instrumental in bounding the recursion depth of our algorithm.


\begin{definition}[Chain length and intersection closure~\cite{DBLP:journals/toct/Guillemot25}]
\label{def:int_closure}
Consider a set system $(\cU,\mathcal{F})$.  For a set $S$ in a set family $\mathcal A$, an \emph{$S$-chain in $\mathcal A$} is a strict chain $S_0\subset S_1\subset\cdots\subset S_q=S$ with every $S_i\in\mathcal A$; its length is $q$.  Let $\ell_{\mathcal A}(S)$ be the maximum length of such a chain.

The \emph{intersection closure} $\intmath(\mathcal F)$ is the smallest family containing all the sets in $\mathcal F$, $\emptyset$, and $\cU$ that is closed under pairwise intersections; equivalently, it consists of all finite intersections of members of $\mathcal F \cup \emptyset $ together with $\cU$.  For the recursive algorithm we also allow the artificial bottom set $\bot=\emptyset$ as an initial seed and set $\ell(\bot)=0$.  For $S\in\intmath(\mathcal F)$, we write $\ell(S)$ for $\ell_{\intmath(\mathcal F)}(S)$.
\end{definition}



\begin{algorithm}[h!]
    \caption{Algorithm to find tuples that \contain all $\ell$-sized minimal set covers of $(\U,\F)$ that extend \family{I}, where $I_i \in \intmath(\F)$ for $i \in [\ell]$}
    \label{alg:find_compact_representations}
    \begin{algorithmic}[1]
        \Procedure{Solve}{$I_1,I_2,\ldots,I_{\ell},\U,\F$} \Comment{ $I_i \in \intmath(\F)$ for $i \in [\ell]$}
        \State $C \gets \bigcup_{i \in [\ell]} I_i$
        \If {$C = \mathcal{U}$} \Comment{All elements are covered}
        \State $F_i \gets \{S \mid S \supseteq I_i \text{ and } S \in \F\}$ for all $i \in [\ell]$ \label{step: check} \Comment{Since $I_i \in \intmath(\F)$, $F_i \neq \emptyset$ for $i \in [\ell]$ }
        \State Delete from each $F_i$ every set that is present in some other $F_j$
        \State \Return $\{\fami{F}\}$
        \EndIf
        \State Select $e \in \mathcal{U} \setminus C$ \Comment{Pick an uncovered element}
        \State $T \gets \emptyset$
        \For {$j \in [\ell]$}
        \State $F_{je} \gets \{S \mid S \supseteq I_j \cup \{e\} \text{ and } S \in \F\}$
        \If {$|F_{je}| \neq 0$}
        \State $I_{je} \gets \bigcap_{f \in F_{je}} f$
        \State $\X{T} \gets \X{T} \cup \textsc{Solve}(I_1, \ldots, I_{je}, \ldots, I_\ell,\U,\F)$
        \EndIf
        \EndFor
        \State \Return $\X{T}$
        \EndProcedure
    \end{algorithmic}
\end{algorithm}

Towards the proof of this theorem we will need to introduce the key concept of an {\it intersection closure}, and use that to prove a couple of technical lemmas.




\begin{proposition}[Theorem 2.2 of \cite{DBLP:journals/toct/Guillemot25}] \label{prop : bounded_length_of_int_closure}
    If $\cF$ is a \dslfree family, then for each set $I \in \operatorname{int}(\mathcal{F})$, it holds that the maximum length of an $I$-chain, $\ell(I) \le d$.
\end{proposition}
\ifthenelse{\equal{\version}{compressed}}
{
}
{
We now present our first key lemma, which links latent covers to the base case of our algorithm. It shows that once an accumulated sequence from the intersection closure covers the universe, the \textsc{Solve} procedure can efficiently construct a valid tuple representing all its extensions.



\begin{restatable}[\app]{lemma}{partialcoverfound}
    \label{lemma : partial_cover_found}
    Consider a set system $(\U, \F)$.
  Suppose that $(I_1,\ldots,I_{\ell})$ is a latent cover where $I_i \in \intmath(\F)$ for each $i \in [\ell]$. Then, procedure \solve\callfamily{I} outputs an $\ell$-tuple $\mathbb{T} = (F_1,\ldots,F_\ell)$ such that
    \begin{itemize}
        \item for each sequence $(S_1,\ldots,S_\ell) \in F_1 \times \ldots \times F_\ell$, we have $\bigcup_{i=1}^\ell S_i = \mathcal{U}$ and the families $\{F_i : ~i\in [\ell]\}$ are pairwise disjoint.
        \item The tuple $\mathbb{T}$ \contains every sequence $(S_1, \dots, S_\ell)$ such that $(S_1, \dots, S_\ell)$ (1) extends $(I_1, \dots, I_\ell)$; and (2) forms a minimal set cover.
        \item The running time of \solve\callfamily{I} is $\Oh(\ell mn)$.
    \end{itemize}
\end{restatable}

    \begin{proof}
    When Line \ref{step: check}
    in \Cref{alg:find_compact_representations} is triggered, we have
    $F_i = \{S \mid S \supseteq I_i \text{ and } S \in \F\}$ for each $i \in [\ell]$.
    Hence, all minimal set covers that minimally extend $(I_1,\ldots,I_{\ell})$ are in $(F_1,\ldots,F_{\ell})$.
    Moreover, for each $(f_1,\ldots,f_\ell) \in F_1\times\cdots\times F_\ell$, it holds that $\bigcup_{i=1}^\ell f_i = \mathcal{U}$ since $(I_1,\ldots,I_{\ell})$ is a latent cover.

    {Let $S^{\star}$ be a minimal set cover of size $\ell$ that minimally extends $(I_1, \ldots, I_{\ell})$, with $S^{\star}_i \supseteq I_i$ for each $i\in [\ell]$. 
    Suppose, for contradiction, that $S^{\star}$ has a set $S \in \mathcal{F}$ such that $S \in F_j \cap F_k$ for distinct indices $j, k \in [\ell]$. 
    This implies $S \supseteq I_j$ and $S \supseteq I_k$.
    Without loss of generality, assume that $S^{\star}_j = S$ (note that due to minimality of $S^\star$, this implies that $S^{\star}_k \ne S$).
    Consider the subfamily $S' = S^{\star} \setminus \{S^{\star}_k\}$. 
    Since $\bigcup_{i=1}^\ell I_i = \mathcal{U}$ and every element in $S'$ is a superset of its corresponding $I$, the union of sets in $S'$ covers $\bigcup_{i=1}^\ell I_i = \mathcal{U}$ (since $S^{\star}_j = S \supseteq I_k$).
    Thus, $S^{\star} \setminus \{S^{\star}_k\}$ is a valid set cover of size $\ell-1$, which contradicts the assumption that $S^{\star}$ is an inclusion-wise minimal set cover of size $\ell$.
    Therefore, no minimal set cover of size $\ell$ can use a set that belongs to multiple $F_i$'s, and such sets can be safely removed to ensure the families are disjoint.}
    
  \textit{Running time analysis.} Every set $S\subseteq \U$ can be represented by an $n$-bit string where the $e$-th bit indicates whether an element $e$ belongs to the set $S$ or not. 
    This takes $\Oh(n)$ time per set.
    Now for each $i \in [\ell], S' \in \F$, we calculate bit strings for $I_i$ and $S'$ and include $S'$ in $F_i$ if all the elements that are present in $I_i$ are also present in $S'$. Now similarly any set family $F \subseteq \F$ can be represented by an $m$-bit string. For each $i \in [\ell]$, we compare $\ell-1$ many bit string and remove sets from $F_i$ that are present in any one of $\{F_1,\ldots,F_\ell\} \setminus \{F_i\}$.
    All of these operations takes $\Oh(\ell mn)$ time.
\end{proof}
}

The next lemma constitutes an essential step toward proving \Cref{theorem:find_compact_representations}.
Let us define a measure $\mu \fami{I} = \ell d - \sum^{\ell}_{i=1} l(I_i)$.

\begin{restatable}[\app]{lemma}{lcompactrepresentationcorrectness}
    \label{lemma:l_compact_representation_correctness}
    Consider a \dslfree set system $(\U, \F)$ and $\intmath(\F)$, the intersection closure of $\F$. 
    For a tuple $(I_1, \ldots, I_{\ell})$ such that $I_i \in \intmath(F)$, for $i\in [\ell]$,
    the procedure \solve\callfamily{I} returns a family of $\ell$-tuples, denoted by $\X{T}$, such that
    \begin{enumerate}
        \item\label{item:covering-packing} for each $\ell$-tuple $\mathbb{T}=(F_1,\ldots,F_{\ell}) \in \X{T}$ and each sequence $(S_1,\ldots,S_\ell) \in F_1 \times \ldots \times F_\ell$, we have $\bigcup_{i=1}^\ell S_i = \mathcal{U}$ and the families $\{F_i :~i\in [\ell]\}$ are pairwise-disjoint.

        \item\label{item:minimal-extension} 
        
        For each minimal set cover \Co{S} whose sequence \family{S} extends \family{I}, there exists a tuple $\mathbb{T} \in \X{T}$ that \contains ~\family{S}. 

        \item\label{item:time-complexity} The size $|\X{T}| \leq \ell^{\mu\fami{I}}$ and the running time of \solve\callfamily{I} is at most $\ell^{\mu\fami{I}+1} \cdot \Oh(nm)$.
    \end{enumerate}
\end{restatable}

\ifthenelse{\equal{\version}{compressed}}
{
}
{
    \begin{proof}
    We prove the statement by induction on the value of $\mu\fami{I}$.

    \noindent\textbf{Base Case:} $\mu\fami{I} = 0$.
    We first show that when $\mu\fami{I}=0$, either the tuple \family{I} already constitutes a latent cover, or no minimal set cover of size exactly $\ell$ exists that minimally extends \family{I}.

    If \family{I} is a latent cover, then by \Cref{lemma : partial_cover_found}, \solve\callfamily{I} correctly outputs a tuple \containing all minimal set covers that minimally extend \family{I} in time $\Oh(\ell mn)$.

    Now, suppose that \family{I} is not a latent cover. Since $I_i \in \intmath(\F)$ for all $i \in [\ell]$, by \Cref{prop : bounded_length_of_int_closure}, we know that $\ell(I_i) \le d$ for all $i \in [\ell]$. Given that $\mu\fami{I} = 0$, it follows that $\ell(I_i) = d$ for all $i \in [\ell]$. Suppose, for contradiction, that there exists a minimal set cover \family{S} that minimally extends \family{I}. Then, there must exist an element $e \in \U \setminus (I_1 \cup \ldots \cup I_{\ell})$. Since \family{S} covers $\U$, $e \in S_j$ for some $j$. This implies that $\ell(S_j) = d+1$ (since $S_j \supsetneq I_j$), contradicting the bound given in \Cref{prop : bounded_length_of_int_closure}. Therefore, no such \family{S} exists in this case, completing the base step.

    \noindent
    \textbf{Induction Hypothesis:}
    Assume that the lemma holds for all instances of \family{I} such that $\mu\fami{I} < t$. That is, for any tuple with a smaller measure, the algorithm correctly outputs all required $\ell$-tuples with the claimed properties and within the stated time bound.

    \smallskip
    \noindent
    \textbf{Inductive Step:}
    We now consider the case $\mu\fami{I} = t > 0$. There are two possibilities to analyze depending on whether \family{I} forms a latent cover or not.

    \medskip
    \noindent
    \textit{Case 1: \family{I} is a latent cover.}
    In this scenario, by \Cref{lemma : partial_cover_found}, \solve\callfamily{I} outputs a tuple satisfying all the desired properties, including the fact that it \contains~all minimal set covers that minimally extend \family{I}, and it does so in time $\Oh(\ell mn)$.

    \medskip
    \noindent
    \textit{Case 2: \family{I} is not a latent cover.}
    Let \family{S} be a minimal set cover that minimally extends \family{I}. Since \family{I} is not a latent cover, there must exist an element $e \in S_j$ for some $j$ that is uncovered by \family{I}, i.e., $e \in \U \setminus (I_1 \cup \ldots \cup I_{\ell})$. \Cref{alg:find_compact_representations} explicitly branches on such an uncovered element $e$ in Line 8.

   Since \family{S} minimally extends \family{I}, we have $I_j \subseteq S_j$. Moreover, since $e \in S_j$, it follows that $I_j \cup \{e\} \subseteq S_j$. We define 
    \[
        I_{je} = \bigcap_{I_j\cup\{e\}\subseteq S \in F} S,
    \]
    that is, $I_{je}$ is the intersection of all sets in $\F$ that contain $I_j \cup \{e\}$. Let $F_{je}$ denote the family of all sets $S \in \F$ such that $I_j \cup \{e\} \subseteq S $.

    Clearly, $F_{je} \neq \emptyset$ since $S_j \in F_{je}$. 
    Moreover, $I_{je} \in \intmath(\F)$ and, by construction, $l(I_j) < l(I_{je})$ because $I_j \subsetneq I_{je}$ (since $e \in I_{je}\setminus I_j$).
    Recall that the measure is $\mu\fami{I} = \ell d - \sum^{\ell}_{i=1} l(I_i)$. Since the chain length $l(\cdot)$ strictly increases for $I_{je}$, the sum subtracted from $\ell d$ strictly increases. Hence, the overall measure strictly decreases: $\mu(I_1, \ldots, I_{je}, \ldots, I_{\ell}) < \mu\fami{I} = t$.
    This allows us to apply the induction hypothesis on the instance $(I_1, \ldots, I_{je}, \ldots, I_{\ell})$.

    By the induction hypothesis, there exists an $\ell$-tuple $\mathbb{T}$ in the family of $\ell$-tuples returned by \solve$(I_1,\ldots,I_{je},\ldots,I_{\ell},\U,\F)$ such that $\mathbb{T}$ \contains~\family{S}. Consequently, due to Line 14 of the algorithm, there must also exist a corresponding $\ell$-tuple $\mathbb{T}$ in the family of $\ell$-tuples returned by \solve$(I_1,\ldots,I_{j},\ldots,I_{\ell},\U,\F)$ that \contains \family{S} as desired.

    Note that we are in the case where \family{I} is not a latent cover. For each $j \in [\ell]$ and each $I_{je}$ constructed from a non-empty $F_{je}$, we have $\ell(I_{je}) > \ell(I_j)$, implying that $\mu(I_1, \ldots, I_{je}, \ldots, I_{\ell}) < t$. Thus, we may again apply the induction hypothesis to all such recursive instances.

    \begin{sloppypar}
        By the induction hypothesis, for every $\ell$-tuple $\mathbb{T}=(F_1,F_2,\ldots,F_{\ell})$ returned by \solve$(I_1,\ldots,I_{je},\ldots,I_{\ell},\U,\F)$, and for every sequence $(X_1,\ldots,X_\ell) \in F_1 \times \cdots \times F_\ell$, we have $\bigcup_{i=1}^{\ell} X_i = \mathcal{U}$. 
        Since in Line 14, the algorithm takes the union over all such recursive calls, it follows that in the final output of \solve\callfamily{I}, for every $\ell$-tuple $\mathbb{T} = (F_1,F_2,\ldots,F_{\ell})$ and every $(X_1,\ldots,X_\ell) \in F_1\times\cdots\times F_\ell$, the union $\bigcup_{i=1}^\ell X_i = \mathcal{U}$ also holds.

    \end{sloppypar}


{\it Bounding the number of $\ell$-tuples.}
  By the induction hypothesis, the number of $\ell$-tuples returned by each call to \solve$(I_1,\ldots,I_{je},\ldots,I_{\ell},\U,\F)$ is at most $\ell^{\mu(I_1,\ldots,I_{je},\ldots,I_{\ell})} \le \ell^{t-1}$, for each $j \in [\ell]$ with $F_{je} \neq \emptyset$. Since there are at most $\ell$ such branches, the total number of $\ell$-tuples returned by \solve\callfamily{I} is at most $\ell^t$.

  \begin{sloppypar}
 {\it Running-time analysis.} Each recursive instance of the algorithm \solve$(I_1,\ldots,I_{je},\ldots,I_{\ell},\U,\F)$ (for each valid $j$) takes time at most $\ell^{(t-1)+1} \cdot \Oh(nm)$, and each such instance outputs at most $\ell^{t-1}$ tuples, each of size at most $\ell \cdot n$. Therefore, the total running time across all recursive branches is bounded by $\ell^{t-1} \cdot \ell \cdot \ell \cdot n = \ell^{t+1} \cdot \Oh(nm)$,
    which establishes the claimed time complexity.
\end{sloppypar}

    This completes the inductive argument, and hence the proof of the lemma. 
\end{proof}
}

We will formally argue that it is sufficient to deal with minimal set covers that extend the tuples of the intersection closure of $\F$. We prove \Cref{theorem:find_compact_representations} in \Cref{sec:app-csc-dslfree}.

\ifthenelse{\equal{\version}{compressed}}
{
}
{
    \begin{proof}[Proof of \Cref{theorem:find_compact_representations}] 
        For each $i \in [k]$, we invoke \solve$(X_1,X_2,\ldots,X_i,\U,\F)$ where $X_j=\bot=\emptyset$ for every $j\in[i]$, and $\ell=i$. Every minimal set cover of size $i$ extends this bottom tuple. Hence, procedure \solve$(X_1,X_2,\ldots,X_i,\U,\F)$ returns a family of tuples $\mathscr{D}_i$ where $|\mathscr{D}_i| \le i^{id}$ that compactly represents all inclusion-wise minimal set covers of size exactly $i$.
    
        We compute $\mathscr{D}$ by taking the union of all $\mathscr{D}_i$ for $i \in [k]$ (that is, $\X{D} = \bigcup_{i\in [k]}\X{D}_i$).
        It compactly represents all inclusion-wise minimal set covers of size at most $k$ in $(\mathcal{U},\mathcal{F})$, satisfying the desired properties in \Cref{theorem:find_compact_representations}.
        Therefore, $|\D|\le k\cdot k^{kd} \le k^{kd+1}$. As every tuple output by the algorithm for all values of $i \in [k]$ satisfies the covering condition, the union also does.
        Moreover, for each $i \in [k]$, computing $\mathscr{D}_i$ requires at most $i^{id+1} \cdot \Oh(nm)$ time. Therefore, the total time to compute $\mathscr{D}$ is bounded by $k^{kd+2} \cdot \Oh(nm)$.
    \end{proof}
}

\subsubsection{Tightness and Impossibility Results}\label{subsubsec:tighteness and impossibility}
In this section, we present a lower-bound construction that establishes both the asymptotic tightness of \Cref{theorem:find_compact_representations} and an unconditional impossibility result for general set systems.
The basic idea of our construction is to obtain an instance where, in any compact representation of the minimal set covers, every tuple is forced to contain only singletons. 
This restriction implies that the number of tuples must scale directly with the number of minimal set covers (accounting for permutations).
Ultimately, by establishing a lower bound on the total number of minimal set covers in this instance, we immediately obtain an unconditional lower bound on the size of the compact representation itself.

\subparagraph{The Construction.}
Let $k, c \in \mathbb{N}$ with $k \ge 2$ and $c \ge 1$.
We define the universe $\mathcal{U}$ as a grid of elements with $k$ rows and $c$ columns:
$
    \mathcal{U} = \{e_{i,j} \mid i \in [k], j \in [c]\}.
$
We define the set family $\mathcal{F}$ such that every set contains exactly one element from each column:
$
    \mathcal{F} = \{ S \subset \mathcal{U} \mid \text{for all } j \in [c], \, |S \cap \{e_{1,j}, e_{2,j}, \dots, e_{k,j}\}| = 1 \}.
$

\ifthenelse{\equal{\version}{compressed}}
{
\begin{restatable}[\app]{claim}{imposscountcovers}
    \label{clm:imposs-count-covers_1}
    Any compact representation of all minimal set covers of size at most $k$ in $(\mathcal{U}, \mathcal{F})$ is of size at least $(k!)^{c-1}$.
\end{restatable}
}
{
    Since there are $k$ choices for each of the $c$ columns, the total number of sets is $|\mathcal{F}| = k^c$. 
    Observe that since every set is of size exactly $c$, every minimal set cover of size at most $k$ is actually a minimum set cover of size exactly $k$.
    Moreover, every minimum set cover is also a partition of the universe: that is, each element of the universe is contained in exactly one $c$-sized set in the set cover.
    Thus, we have the following.
    
    \begin{claim}
        \label{clm:imposs-singletons-in-D}
        In any compact representation of all minimal set covers of size at most $k$ in $(\mathcal{U}, \mathcal{F})$, every sub-family of sets present in any $k$-tuple must be a singleton.
        That is, for any $k$-tuple $\mathbb{T} = (F_1, \dots, F_k)$ in the representation, $|F_i| = 1$ for all $i \in [k]$.
    \end{claim}
    \begin{claimproof}
        By the covering property of a compact representation, every combination of sets $(S_1, \dots, S_k) \in F_1 \times \dots \times F_k$ must cover $\mathcal{U}$. 
    
        Fix arbitrary sets $S_2 \in F_2, \dots, S_k \in F_k$. Since each set covers exactly $c$ elements, their union covers at most $(k-1)c$ elements. 
        This leaves at least $c$ elements of $\mathcal{U}$ uncovered (specifically, exactly one element in each column). 
    
        For the tuple to be valid, \emph{every} set $S_1 \in F_1$ must cover exactly these remaining $c$ elements. 
        However, by the definition of $\mathcal{F}$, there is exactly one set in $\mathcal{F}$ that contains those specific $c$ elements. 
        Therefore, $F_1$ is a singleton. 
        By symmetry, this holds for all $F_i$.
    \end{claimproof}
    
    Let $\mathscr{D}$ denote \emph{any} compact representation of all minimal set covers of size at most $k$ in $(\mathcal{U}, \mathcal{F})$.
    Consider any minimal set cover $(S_1, \dots, S_k)$ in the set system.
    Observe that such a set cover always exists; for example, one can choose $S_i$ to be exactly the elements in the $i$-th row of the grid.
    By definition, there is some $k$-tuple in $\mathscr{D}$ that \contains this set cover.
    Moreover, by \Cref{clm:imposs-singletons-in-D}, this is the \emph{unique} set cover that this tuple \contains.
    Thus, the number of $k$-tuples in $\mathscr{D}$ is tightly lower-bounded by the number of minimal set covers of size $k$ in $(\mathcal{U}, \mathcal{F})$.
    
    \begin{claim}
        $|\mathscr{D}| \ge \text{(number of minimal set covers of size $k$)}$.
    \end{claim}
    
    To count the number of minimal set covers of size $k$, we count the number of ways to partition the grid into $k$ disjoint sets of size $c$. 
    We process the universe column by column:
    \begin{itemize}
        \item Column 1: There are $k$ elements in the first column to be distributed among $k$ \emph{unlabelled} sets. 
        There is exactly $1$ way to initialize this partition.
        \item Columns 2 to $c$: For each subsequent column, we must assign its $k$ elements to the $k$ established sets. 
        There are exactly $k!$ ways to distinctly map the elements of the column to the $k$ sets.
    \end{itemize}
    Since the choices for columns $2$ to $c$ are completely independent, we have the following.
    \begin{claim}
        \label{clm:imposs-count-covers}
        The number of minimal set covers of size $k$ in $(\mathcal{U}, \mathcal{F})$ is $(k!)^{c-1}$.
        Consequently, $|\mathscr{D}| \ge (k!)^{c-1}$.
    \end{claim}
}

Next, we derive the consequences of this construction by choosing appropriate values of $c$.
First, we prove the tightness of \Cref{theorem:find_compact_representations} in \dslfree set systems.

\begin{restatable}[\app]{theorem}{tightnessdslfree}
    \label{thm:tightness-dsl-free}
    When $c=d-2$, the set system $(\mathcal{U}, \mathcal{F})$ is \dslfree. 
    Any compact representation of all minimal set covers of size at most $k$ has size at least $(k!)^{d-3} = k^{\Omega(kd)}$.
\end{restatable}

\ifthenelse{\equal{\version}{compressed}}
{
}
{
\begin{proof}
    When $c=d-2$, all sets in $\mathcal{F}$ are of size $d-2$.
    By definition, a $d$-semi-ladder requires $d$ sets $S_1, \dots, S_d$ and $d$ elements $x_1, \dots, x_d$ such that $x_i \in S_j$ for $j > i$, and $x_i \notin S_i$. 
    Specifically, the set $S_d$ must contain the $d-1$ elements $x_1, \dots, x_{d-1}$. 
    Since no set in $\mathcal{F}$ can contain two elements from the same column, these $d-1$ elements must belong to distinct columns. 
    However, the grid only has $(d-2)$ columns. Therefore, such a set $S_d$ cannot exist.
    Thus, there is no $d$-semi-ladder in $\mathcal{F}$.
    
\end{proof}
}

Next, we prove an impossibility result on obtaining small compact representations. 
The $\textsf{W[2]}$-hardness of \textsc{Set Cover} rules out an \fpt algorithm that computes an \fpt sized compact representation.
However, it does not preclude the existence of such small compact representations.
We prove that for general instances where the semi-ladder index is unbounded, the size of any compact representation cannot be bounded by any \fpt function.

\begin{restatable}[\app]{theorem}{insertlabelhere}
    \label{thm:tightness-of-compact-rep-in-general-set-systems}
    For any function $f$ and any constant $q$, there exists an integer $k$ and a set system $(\mathcal{U}, \mathcal{F})$ such that any compact representation of all minimal set covers of size at most $k$ has size strictly greater than $f(k) \cdot (|\mathcal{U}| + |\mathcal{F}|)^q$.
\end{restatable}
\ifthenelse{\equal{\version}{compressed}}
{
}
{
\begin{proof}
    We choose $k$ to be any integer sufficiently large such that $\log_k(k!) > q$.
    Consider our construction $(\mathcal{U}, \mathcal{F})$ with $k$ rows and $c$ columns, where $c$ is an integer we will choose later.
    
    Recall that $|\mathcal{U}| = kc$ and $|\mathcal{F}| = k^c$.
    For any positive $c$, we have $kc \le k^c$, and thus $|\mathcal{U}| + |\mathcal{F}| \le 2k^c$. 
    
    Let $\mathscr{D}$ denote an arbitrary compact representation of all minimal set covers of size at most $k$.
    By \Cref{clm:imposs-count-covers}, we have $|\mathscr{D}| \ge (k!)^{c-1}$.
    Therefore, for the size to exceed the stated \fpt bound, it suffices to choose $c$ such that:
    \[
        |\mathscr{D}| \ge (k!)^{c-1} > f(k) \cdot (2k^c)^q \ge f(k) \cdot (|\mathcal{U}| + |\mathcal{F}|)^q.
    \]
    Taking the logarithm base $k$ yields:
    \[
        (c-1) \log_k(k!) > \log_k(f(k)) + q\log_k(2) + c \cdot q.
    \]
    Rearranging the terms gives:
    \[
        c \cdot \left( \log_k(k!) - q \right) > \log_k(f(k)) + q\log_k(2) + \log_k(k!).
    \]
    Since we chose $k$ to be sufficiently large, the coefficient $(\log_k(k!) - q)$ is a positive constant. 
    Moreover, since $k$ and $q$ are now fixed, the entire right-hand side is a constant.
    Finally, we set $c$ sufficiently large so that the inequality holds, and thus we have the theorem.
\end{proof}
}

\subsection{Applications of compactly representing all minimal set covers}
\label{sec:apps-of-compact-representations}

Now that we have a method to find the compact representations of all minimal set covers in \dslfree set systems, we can leverage it to solve \Csc efficiently, \Cref{thm:connected-set-cover}, and as a consequence prove the following result on \cds. 


\hide{\begin{restatable}[\app]{theorem}{cscalgo}
    \label{thm:connected-set-cover} There exists an algorithm for \Csc on \dslfree set systems that runs in $k^{kd+2}\cdot 2^k \cdot n^{\Oh(1)}$ time and uses polynomial space.
\end{restatable}}

\ifthenelse{\equal{\version}{compressed}}
{
}
{
    \begin{proof}
        Let $\mathscr{D}$ be the family of tuples returned by \Cref{theorem:find_compact_representations} in $k^{kd+2}\cdot \Oh(nm)$ time.
        We would like to invoke the \gst algorithm given by \Cref{prop : group steiner tree algo} to find a solution.
        For each tuple $\mathbb{T}=\{F_1,F_2,\ldots,F_\ell\}$, we define the groups to be $F_i$ for each $i \in [\ell]$. Now we will run the \gst algorithm given by \Cref{prop : group steiner tree algo} on the graph $\hat{G}$ to find a solution. We will call such an instance the \gst instance on tuple $\mathbb{T}$.
        Any solution $T=\{S_1,\ldots,S_\ell\}$ to \gst instance on any tuple is also a solution for the \Csc instance as at least one set of each family must be taken in $T$, which ensures that all elements are covered by condition~\ref{compact_rep:covering_condition}.
        Also, as $T$ is a solution of the \gst instance, $T$ is connected.
        What remains to be shown is that, if the \Csc instance parameterized by solution size is a \yes instance, then there exists a tuple such that \Cref{prop : group steiner tree algo} returns a solution of size at most $k$.
        Let $D^\star$ be an optimal solution of the \Csc instance of size at most $k$. Let $D'\subseteq D^\star$ be a minimal solution. We know by \Cref{def:compact_representations} that there exists a tuple $\mathbb{T}$ in $\mathscr{D}$ that \contains $D'$ as $|D'| \le k$. Since $D^\star$ is a solution for the \gst instance on the tuple $\mathbb{T}$, there exists a solution returned by \Cref{prop : group steiner tree algo} of size at most $k$.
        
        \noindent
        \textit{Running time analysis.} The family of tuples $\mathscr{D}$ of size at most $k^{kd+1}$ can be computed in $k^{kd+2}\cdot \Oh(nm)$ time (\Cref{theorem:find_compact_representations}).
        Now, for each tuple, the \gst algorithm (\Cref{prop : group steiner tree algo}) takes $2^\ell \cdot n^{\Oh(1)}$ time, the entire algorithm runs in $k^{kd+2}\cdot 2^k \cdot n^{\Oh(1)}$ time.
    \end{proof}
}


\begin{restatable}[\app]{corollary}{csccds}
    \label{corrolary:conn-dominating-set}
    There exists an algorithm for \cds on \dslfree graphs that runs in $k^{kd+2}\cdot 2^k \cdot n^{\Oh(1)}$ time and in polynomial space.
\end{restatable}
\ifthenelse{\equal{\version}{compressed}}
{
}
{
    \begin{proof}
        Given an instance $(\calG,k)$ of \cds, we model it as an instance $(({\cal U, \cal F}),\hat{G})$ of \Csc as follows. We set $\cU = V(\calG)$. The family $\cF$ consists of the closed neighborhoods of all the vertices in $\calG$.
    Two sets $A,B \in \cF$ have an edge between them in $\hat{G}$ if the corresponding vertices share an edge in $\calG$. It is easy to show that the instance $(\calG,k)$ of \cds has a solution of size at most $k$ if and only if the $(({\cal U, \cal F}),\hat{G})$ instance of \Csc has a solution of size at most $k$.
\end{proof}
}

\newcommand{\rbds}{\textsc{Red-Blue Dominating Set}\xspace}
\section{Cores for Dominators and Domination Requirements}
\label{sec:cores}


In this section, we develop structural tools for the lossy compression.  We reinterpret \ds as \rbds, separating possible dominators from vertices that must be dominated.  The two objects below are parameter-$k$ cores: all quantifiers over dominating sets are restricted to inclusion-wise minimal dominating sets of size at most $k$.

\begin{definition}[$k$-dominator core]\label{def : Dominator Core}
Let $G$ be a graph and $k\in\mathbb N$.  A family $\mathcal C=\{C_1,\ldots,C_t\}$ of pairwise disjoint subsets of $V(G)$ is a \emph{$k$-dominator core} of $G$ if the following hold for every inclusion-wise minimal dominating set $D$ of $G$ with $|D|\le k$.
\begin{enumerate}
    \item \emph{(Enclosure)} $D\subseteq \bigcup_{i\in[t]} C_i$.
    \item \emph{(Unique contribution)} $|D\cap C_i|\le 1$ for every $i\in[t]$.
    \item \emph{(Replacement)} If $v\in D\cap C_i$ and $v'\in C_i$, then $(D\setminus\{v\})\cup\{v'\}$ is a dominating set of $G$.
\end{enumerate}
We call $|\mathcal C|$ the size of the core.
\end{definition}

Given a \dslfree graph $G$, our goal is to compute a $k$-grouped domination core of size at most $k^d$ and a $k$-dominator core of size $\Oh(k^{d^2})$.




\begin{definition}[Domination of Sets and Families] \label{def : dominating set and families}
    In the graph $G$, a vertex $u$ is said to \emph{dominate} a subset $S\sse V(G)$ if $u$ dominates every vertex in $S$.
    A family of vertex subsets $\F \subseteq 2^{V(G)}$ is said to be dominated by a set $D\subseteq V(G)$ if for each set $B \in \F$, there exists a vertex $v \in D$ that dominates $B$.
\end{definition}

The next definition groups vertices into subsets that must be collectively dominated. 


\begin{definition}[$k$-grouped domination core]\label{def : Grouped Domination Core}
Let $G$ be a graph and $k\in\mathbb N$.  A family $\mathcal G$ of pairwise disjoint subsets of $V(G)$ is a \emph{$k$-grouped domination core} of $G$ if the following hold.
\begin{enumerate}
    \item \emph{(Collective domination)} For every inclusion-wise minimal dominating set $D$ of $G$ with $|D|\le k$ and every batch $B\in\mathcal G$, some vertex of $D$ dominates all vertices of $B$.  We refer to each set $B \in {\mathcal G}$ as a \emph{batch}.
    \item \emph{(Preservation)} Every set $D\subseteq V(G)$ that dominates every batch in $\mathcal G$ is a dominating set of $G$.
\end{enumerate}
\end{definition}



Next, we establish the necessary structural properties of \dslfree graphs and the required grouping operation.
The \textsc{Group} operation transforms a tuple $(\U, \F, f, g, S)$ into $(\U', \F', f', g')$ by replacing all vertices in $S$ with a single representative. 
This preserves \dslfree structure and maintains domination and grouping properties. 
In particular, any dominating set in $( \U',\F', f', g')$ corresponds naturally to one in the original graph. 
We use this operation as a key subroutine in our algorithm; its formal description follows.



\ifthenelse{\equal{\version}{compressed}}
{
}
{
The following result will be used to argue the base case in the proof of \Cref{thm:existence-domination-dominator-core}.
Note that if a vertex appears in all closed neighborhoods of $G$, then that vertex itself is a solution; thus, the \cds instance parameterized by solution size becomes polynomial-time solvable and the vertex can be returned as a singleton in dominator core and grouped domination core if $k\ge 1$. Henceforth, we shall assume that the size of a minimum connected dominating set of $G$ is strictly greater than $1$.

\begin{proposition}[\cite{DBLP:journals/toct/Guillemot25}]
    \label{prop : bounding number of neighbourhoods}
    Let $\mathcal{F}$ be a \dslfree family of sets over a universe of size $n$ such that no element of the universe appears in all the sets of $\mathcal{F}$. Then, the number of distinct sets in $\mathcal{F}$ is bounded by $\Oh(n^d)$. 
\end{proposition}
}

\subparagraph*{Operation \textsc{Group}.}
A \emph{relevant tuple} for a graph $G$ is a tuple $(\U,\F,f,g)$ where $(\U,\F)$ is a set system, $f:\F\to 2^{V(G)}$, and $g:\U\to 2^{V(G)}$.  The operation \textsc{Group} takes a relevant tuple and a subset $S\subseteq\U$ and replaces all elements of $S$ by one new representative $v_S$.  Formally, $\U'=(\U\setminus S)\cup\{v_S\}$.  Each set $A\in\F$ is mapped to
\[
\sigma(A)=
\begin{cases}
(A\setminus S)\cup\{v_S\},&\text{if }S\subseteq A,\\
A\setminus S,&\text{otherwise.}
\end{cases}
\]
After duplicate images are identified, let $\F'=\{\sigma(A):A\in\F\}$.  For $B\in\F'$, define $f'(B)=\bigcup_{A\in\F:\sigma(A)=B} f(A)$.  Finally, set $g'(u)=g(u)$ for $u\in\U\setminus S$ and $g'(v_S)=\bigcup_{s\in S}g(s)$.  If $|S|\ge 2$, then $|\U'|<|\U|$.

\begin{proposition}[Lemma 3.3 in \cite{DBLP:journals/toct/Guillemot25}]
    \label{prop:semiladder closed grouping}
    Let $(\U,\F)$ be a set system where $\F$ is \dslfree. Then, for any subset $S\subseteq\U$, the set family $(\U',\F')$ returned by \textsc{Group}$(\U,\F,f,g,S)$ is \dslfree.
\end{proposition}


\subparagraph*{The \textsc{FindCore} algorithm.}
We introduce the \textsc{FindCore} algorithm (see \Cref{algo:findcore}), which constructs both $k$-dominator cores and $k$-grouped domination cores \Ma{motivated by \cite{DBLP:journals/toct/Guillemot25}}. 
Intuitively, \textsc{FindCore} recursively compresses the universe while preserving domination structure. 
We complete the analysis in \Cref{sec:app-cores} showing that it yields \Cref{thm:existence-domination-dominator-core}.


\ifthenelse{\equal{\version}{compressed}}
{
}
{
The \emph{image set} of a function $h: A \to B$ denotes the subset of elements of $B$ that have a pre-image in $A$, and is denoted by $f(A)$.
}



\begin{algorithm}[h!]
    \caption{Finding Dominator Core and Grouped Domination Core}
    \label{algo:findcore}
     
    \begin{algorithmic}[1]
        \Procedure{FindCore}{$\mathcal{U}, \F, f, g$}\\
        \hspace*{\algorithmicindent} \textbf{Input}: {Relevant tuple $( \mathcal{U}, \mathcal{F}, f, g)$ of a graph $G$. }\\
        \hspace*{\algorithmicindent} \textbf{Output}: {$k$-dominator core and $k$-grouped domination core}
        \If{$|\mathcal{U}| \le k^d$}
        \State \Return $(f(\F), g(\U))$ \Comment{the dominator core and
        the grouped domination core.}
        \EndIf
        \State Let $\intmath(\mathcal{F})$ be the intersection closure of $\mathcal{F}$.
        \State Let $S \in \intmath(\mathcal{F})$ which satisfies $k^{\ell(S)} < |S|$ and $\ell(S) \leq \ell(X)$ for every $X \in \intmath(\mathcal{F})$ with $k^{\ell(X)} < |X|$.
        \If{no such set $S$ exists}
        \State \label{base case of induction algo: findcore} \Return $(f(\F), g(\U))$ \Comment{the dominator core and
        the grouped domination core.}
        \Else 
        \State $( \mathcal{U}',\mathcal{F}', f', g') \gets \textsc{Group}(\mathcal{U}, \mathcal{F}, f, g, S)$
        \State \label{algo:recursive-call} \Return \Call{FindCore}{$ \mathcal{U}',\mathcal{F}', f', g'$}
        \EndIf
        \EndProcedure
    \end{algorithmic}
\end{algorithm}


\ifthenelse{\equal{\version}{compressed}}
{
}
{
Using the reasoning of \cite{DBLP:journals/toct/Guillemot25}, we can also show the following observation, which will be used in the proof of base case of \Cref{lemma:inductive-core}, the main technical result of this section. For the sake of completeness, we provide the proof of this observation here.

\begin{observation}[\cite{DBLP:journals/toct/Guillemot25}\app]\ma{why does it have reference and also a proof? So: The proof is heavily based on the reasoning of the cited paper}
\label{obs:existance of grouping}
    For the \dslfree set family $\Co{F}$, let $\intmath(\mathcal{F})$ denote the intersection closure of $\mathcal{F}$. If there exists a set $S' \in \intmath(\mathcal{F})$ with $k^{\ell(S')} < |S'|$, then Line~\ref{base case of induction algo: findcore} of \textsc{FindCore}$( \mathcal{U}, \mathcal{F}, f, g, S)$ is never executed.
\end{observation}
}
\ifthenelse{\equal{\version}{compressed}}
{
}
{
    \begin{proof}
        Assume for the sake of contradiction that there exists a set $S' \in \intmath(\mathcal{F})$ with $|S'| > k^{\ell(S')}$ and line~7 of \textsc{FindCore}$(\mathcal{F}, \mathcal{U}, f, g)$ is executed. Let $\hat{F}$ be the family of sets in $\intmath(\mathcal{F})$ that have size greater than $k^{\ell(S')}$ and let $i$ be the minimum value of $\ell(S'')$ for all $S'' \in \hat{F}$. We know that $\hat{F}$ is non-empty as $S' \in \hat{F}$. Let $S$ be the set in $\hat{F}$ such that $\ell(S) = i$. By the definition of $i$, there is no set $S'' \in \intmath(\mathcal{F})$ with $\ell(S'') < i$ and $|S''| > k^{\ell(S'')}$. Hence, line~7 of \textsc{FindCore}$(\mathcal{F}, \mathcal{U}, f, g)$ is never executed, which is a contradiction
    \end{proof}
}

\ifthenelse{\equal{\version}{compressed}}
{
}
{
The following lemma serves as the base case for the inductive proof of \Cref{thm:existence-domination-dominator-core}. It establishes that if the universe $\mathcal{U}$ is sufficiently small, then the image sets of $f$ and $g$ already form the required dominator and grouped domination cores, respectively.





\begin{lemma}\label{lemma:core base case}Consider a relevant tuple $(\U, \F, f,g)$ for a graph $G$.
    Suppose that the set family $\F$ is \dslfree.\hide{ with the corresponding universe $\U$.} Moreover, suppose that the image set of $f$ forms a $k$-dominator core of $G$, the image set of $g$ forms a $k$-grouped domination core of $G$, and $|\mathcal{U}| \le k^d$. Then, \textsc{FindCore}$(\mathcal{U}, \mathcal{F}, f, g)$ outputs a $k$-grouped domination core of size at most $k^d$ and a $k$-dominator core of size at most $\Oh(k^{d^2})$.
    \hide{\begin{itemize}
        \item \textsc{FindCore}$(\mathcal{U},\mathcal{F},  f, g)$ outputs a $k$-grouped domination core of size at most $k^d$; and
        \item \textsc{FindCore}$(\mathcal{U}, \mathcal{F}, f, g)$ outputs a dominator core of size at most $\Oh(k^{d^2})$.
    \end{itemize}}
\end{lemma}
\begin{proof}
    The statement of \Cref{lemma:core base case} follows directly from \Cref{prop : bounding number of neighbourhoods}, together with the assumption that the image of $f$ forms a $k$-dominator core and the image of $g$ forms a
    $k$-grouped domination core.
\end{proof}
}

\ifthenelse{\equal{\version}{compressed}}
{
}
{
We begin by defining the problem of \rbds, which will be used as a tool in the proof of the main lemma of this section. In \rbds, we are given a bipartite graph $G = (R \uplus B, E)$ and an integer $k$, and the goal is to determine whether there exists a subset $S \subseteq R$ of size at most $k$ such that every vertex in $B$ is adjacent to at least one vertex in $S$. We will call $R$ as the red set and $B$ as the blue set of $G$.


Let $\hat{G}= (R\uplus B, E)$ be the graph constructed as follows. We create two copies of the vertex set $V(G)$, denoted by $R$ and $B$. We add edge $(x,y)$ to $E$ if the corresponding vertex of $x$ in $R$ is adjacent to the corresponding vertex of $y$ in $B$ in the original graph $G$ or if $x$ and $y$ correspond to the same vertex in $G$.

Now, for a relevant tuple $(\U, \F,  f, g, S)$, we define the \emph{relevant red-blue dominating set} instance as follows.
$V(G') = R \uplus B $ where $R$ and $B$ are copies of vertices in $V(G)$. Moreover, the edges $E(G')$ are defined as follows: for each subset $S'\in \mathcal{F}$, we consider each element $v \in f(S')$. The corresponding vertex $v\in R$ is adjacent to the corresponding vertices of $B$ in the set $\cup_{s\in S'}g(s) $.
Both of these constructions will be used to design an invariant which will be useful to prove the existence of a bounded dominator core and a grouped domination core and thus help in proving the correctness using strong induction. One will give an handle to the original graph $G$ and the other will give a handle to the properties of the relevant tuple $(\U, \F,  f, g)$.




\begin{observation}\label{obs:rbds-ds}
    A set $S \subseteq V(G)$ is a dominating set of $G$ if and only if the corresponding vertices of $S$ in $R$ form a red-blue dominating set of $\hat{G}$.
\end{observation}
This preserves the dominating sets of $G$. In particular, it preserves all the minimal dominating sets of $G$. This is an observation that will be used in the following pivotal lemma.


\begin{restatable}[\app]{lemma}{inductivecore}
\label{lemma:inductive-core}
Let $(\mathcal{U}, \mathcal{F}, f, g)$ be a relevant tuple of a graph $G$ and a subset $S \subseteq V(G)$. Suppose the following conditions are satisfied:
\begin{enumerate}
    \item $\mathcal{F}$ is a $d$-semiladder-free set family.
    \item  The image of $f$ constitutes a $k$-dominator core of $G$.
    \item The image of $g$ constitutes a $k$-grouped domination core of $G$.
    \item The graph $G'(\mathcal{F}, \mathcal{U}, f, g) = (R' \uplus B', E')$ is a subgraph of $\hat{G}$ such that their blue vertex sets are identical. Furthermore, any set $S' \subseteq V(\hat{G})$ of size at most $k$ is a red-blue dominating set of $\hat{G}$ if and only if it is a red-blue dominating set of $G'(\mathcal{F}, \mathcal{U}, f, g)$ of size at most $k$.
\end{enumerate}
Then, the algorithm \textsc{FindCore} outputs a $k$-dominator core and a $k$-grouped domination core of size at most $\Oh(k^{d^2})$ and $k^d$, respectively.
\end{restatable}\ma{So: See once}

    
    
\ifthenelse{\equal{\version}{compressed}}
{
}
{
    \begin{proof}
        We proceed by induction on $|\mathcal{U}|$.
        We say that a vertex partially has edges to a set $S$ if it {has an edge to some}\ma{shares an edge with at least one vertex in $S$..?} but not all\ma{why can't it be "all".? . So: used as a notation in proof. Actually it can be 0 for the case.. we can give a different name} vertices in $S$.
        
        \noindent\textbf{Base case.}
        When $|\mathcal{U}| \le k^d$, the claim follows directly from \Cref{lemma:core base case}.
    
        \noindent\textbf{Induction hypothesis.}
        Assume that the statement holds for all instances of \textsc{FindCore} with $|\mathcal{U}| < t$.
    
        \noindent\textbf{Inductive step.}
        Consider $|\mathcal{U}| = t > k^d$. Let $\intmath(\mathcal{F})$ be the intersection closure of $\mathcal{F}$.
    
        \textit{Case 1.} Suppose line~7 of \textsc{FindCore}$(\mathcal{F}, \mathcal{U}, f, g)$ is triggered.
        Then $|\mathcal{U}| \le k^d$, since otherwise $\mathcal{U} \in \intmath(\Co{F})$ and by \Cref{obs:existance of grouping}, we reach a contradiction. Thus, by \Cref{lemma:core base case}, we obtain both the dominator and grouped domination cores.
    
        \textit{Case 2.} Let $S = \{v_{i_1}, \ldots, v_{i_q}\}$ be the set in $\intmath(\Co{F})$ satisfying
        \begin{equation}
            \label{eq:condition for S}
            \ell(S) = i, \quad |S| > k^{\ell(S)}, \quad \text{and no } S' \in \intmath(\Co{F}) \text{ with } \ell(S') < \ell(S) \text{ has } |S'| > k^{\ell(S')}.
        \end{equation}
        Thus, any intersection with $S$ by a set in $\mathcal{F}$ has size at most $k^{i-1}$. Also in $G'(\F,\U,f,g)$, no vertex in $R$ has an edge
        partially to any batch of the image set of $g$. Let $D'$ be any dominating set of $G'(\F,\U,f,g)$ of size at most $k$.
        By the assumption of the lemma, $D'$ is also a solution of $\hat{G}$.
        As the sets of $k$-grouped domination core (image sets of $g$) are disjoint equation~\ref{eq:condition for S} ensures that no $k$ vertices other than the one that has edges to all the vertices in $g(v_{i_1}) \cup \cdots \cup g(v_{i_q})$ can dominate all of $g(v_{i_1}) \cup \cdots \cup g(v_{i_q})$ in $G'(\F,\U,f,g)$. Therefore, some vertex in $G'(\F,\U,f,g)$ must dominate this entire group. We perform the \textsc{Group} operation on $(\F,\U,f,g,S)$ to obtain a new instance $(\F', \U', f', g')$.
    
        If the four conditions of \Cref{lemma:inductive-core} hold for \textsc{FindCore}$(\mathcal{F}', \mathcal{U}', f', g')$, then by the induction hypothesis, \textsc{FindCore}$(\mathcal{F}', \mathcal{U}', f', g')$ outputs a $k$-grouped domination core of size at most $k^d$ and a $k$-dominator core of size at most $\Oh(k^{d^2})$.
        Hence, by line~\ref{algo:recursive-call} of \textsc{FindCore}$(\mathcal{F}, \mathcal{U}, f, g)$, the required cores are obtained.
        Next, we will prove that the four conditions of \Cref{lemma:inductive-core} also hold for the instance \textsc{FindCore}$(\mathcal{F}', \mathcal{U}', f', g')$. Specifically, we argue as follows.
    
        Due to \Cref{prop:semiladder closed grouping}, we know that the \Sc instance $(\F', \U')$ is $d$-\semiladderfree. This proves the condition~1 of \Cref{lemma:inductive-core}.
    
    
        Let $D$ be a minimal dominating set of $G'(\F,\U,f,g)$ of size at most $k$.
        Hence, there exists a vertex in $G'(\F,\U,f,g)$ that dominates all the vertices in $\bigcup_{j=1}^q g(v_{i_j})$. Thus, all other vertices that have edges to a strict subset of $\bigcup_{j=1}^q g(v_{i_j})$ can have their edges removed from $\bigcup_{j=1}^q g(v_{i_j})$ thus forming $G'(\F',\U',f',g')$. Any dominating set of $G'(\F,\U,f,g)$ is also a dominating set of $G'(\F',\U',f',g')$. Moreover, any dominating set of $G'(\F',\U',f',g')$ is a dominating set of $\hat{G}$ because $G'(\F',\U',f',g')$ is a subgraph of $G'(\F,\U,f,g)$ and hence also a subgraph of $G''$ such that it contains all the vertices in the blue set.
        Thus, by condition~4, any set $S' \subseteq V(G)$ of size at most $k$ is a dominating set of $G'(\F',\U',f',g')$ if and only if it is a dominating set of $G'(\F,\U,f,g)$. Also, by \Cref{obs:rbds-ds}, $S$ is a dominating set of $G'(\F',\U',f',g')$ of size at most $k$ if and only if it is a dominating set of $\hat{G}$ of size at most $k$.
        Thus, Condition~4 of \textsc{FindCore}$(\F',\U',f',g')$ is therefore true.

        We observed that there exists a vertex in the dominating set that dominated all the vertices of $g(v_{i_1}) \cup \ldots \cup g(v_{i_q})$. Next, we would like to observe that sets in $\F$ only get merged into one as elements from the sets get deleted. Those are exactly the vertices that partially dominate vertices of $g(v_{i_1}) \cup \ldots \cup g(v_{i_q})$. Therefore, any minimal dominating set $D$ of $G'(\F,\U,f,g)$ must take at most one vertex from $f(S_1) \cup \ldots \cup f(S_{\ell})$, where $S_1, \dots, S_\ell$ are the sets in $\F$ that got merged together in a set $S''$; otherwise, we could delete an element from $D$ contradicting that $D$ was a minimal dominating set. As the vertices in a set $Q$ of the image set of $f'$ have edges to the same vertices in $G'(\F',\U',f',g')$, let $v$ and $v'$ be any two vertices of such a set $Q$. If $v \in D$, then $(D \setminus \{v\}) \cup \{v'\}$ is also a dominating set of $\hat{G}$, and the corresponding vertices also form a dominating set of $G$. All the sets in the image set of $f'$ are disjoint because some sets which were disjoint under $f$ got merged together. Since we are not deleting any vertex in the image set of $f'$ that has an edge (where all of its neighbours are not dominated by some other vertex), all minimal dominating sets are contained within the image set of $f'$. Hence, the image set of $f'$ is a $k$-dominator core of $G$. Thus, condition~2\ma{use ref, So: How to do?} of \Cref{lemma:inductive-core} for the instance \textsc{FindCore}$(\F',\U',f',g')$ is true.
    
        Since the image set of $g$ is a $k$-grouped domination core of $G$, all the sets in the image set of $g$ are disjoint. The merge operation maintains this disjointness. Thus, the sets in the image set of $g$ are disjoint. By the arguments used previously, there must be a vertex that dominates all the vertices in the newly grouped set. The same is true for the rest of the sets by the assumption that the image set of $g$ is a grouped domination core of $G$. Thus, the image set of $g'$ is a grouped domination core of $G$. Thus, we have proved condition~3 of \textsc{FindCore}$(\F',\U',f',g')$.
    
        \sloppypar{Since $|\U'| < |\U|$, as stated earlier we can apply the induction hypothesis to conclude that \textsc{FindCore}$(\mathcal{F}', \mathcal{U}', f', g')$ constructs a $k$-grouped domination core of size at most $k^{d}$ and a $k$-dominator core of size at most $k^{d^2}$.}
    \end{proof}
}

Armed with the proof for~\Cref{lemma:inductive-core}, the proof for the following theorem is immediate and we move it to \Cref{sec:app-cores}.
}


\ifthenelse{\equal{\version}{compressed}}
{
}
{
    \begin{proof}[Proof of \Cref{thm:existence-domination-dominator-core}]\ma{Move to app (later)}
    We now describe how to invoke \textsc{FindCore} to obtain the required cores.
    We initialize the instance $(\mathcal{F}, \mathcal{U}, f, g)$ as follows: 
    \[ \U=V(G), \quad \F = \{N_G[v] \mid v\in V(G)\}, \quad f(N_G[v]) = \{v\}, \textrm{ and } g(v) = \{v\}\]
    
    
    We verify that this initialization satisfies all four conditions of \Cref{lemma:inductive-core}:
    \begin{itemize}
        \item By \Cref{def:graph-semiladder}, the family $\mathcal{F}$ consisting of closed neighborhoods in a \dslfree graph $G$.  
        \item The image set of $f$, $f(\F)$, forms a $k$-dominator core of $G$ since it consists of singletons $\{v\}$ for each $v \in V(G)$; hence, every minimal dominating set is trivially contained in this family and intersects each set in at most one vertex.
        \item The image set of $g$ forms a grouped domination core of $G$ for the same reason, as any minimal dominating set dominates every singleton $\{v\}$.
        \item The auxiliary graph $G'(\F,\U,f,g)$ coincides with $\hat{G}$, since for every $S = N_G[v] \in \mathcal{F}$, all vertices in $f(S) = \{v\}$ are adjacent to all vertices in $g(v) = \{v\}$ by construction. Therefore, any dominating set of $G$ of size at most $k$ is also a dominating set of $G'(\F,\U,f,g)$, and vice versa.
    \end{itemize}
    
    Hence, all the preconditions of \Cref{lemma:inductive-core} hold for this instance.
    Applying \textsc{FindCore}$(\mathcal{F}, \mathcal{U}, f, g)$ yields the desired output. 
    This completes the proof of \Cref{thm:existence-domination-dominator-core}.
    \end{proof}
}
%
%


\section{Approximate compression for \cds on \dslfree graphs}
\label{sec:lossy-kernel}

We refer to~\Cref{def:approxreduction,def:approxcompression} for the formal definition of approximate compression.
Using \Cref{thm:existence-domination-dominator-core}, we obtain a $k$-dominator core $\C$ and a $k$-grouped domination core $\G$ of size at most $\Oh(k^{d^2})$ and $k^d$, respectively. 
We call the sets in $\C$ \blocks, and let $F \subseteq \G$ denote the non-singleton sets.

Given a graph $G$, we construct a graph $\hG$ by adding to $G$ a set of vertices and edges: for each $S \in F$, we add a vertex $v_S$ and edges $(v_S,x)$ for every $x \in V(G)$ that is adjacent to all vertices in $S$.
Let $H=\{v_S \mid S \in F\}$ denote the set of \emph{heavy} vertices. The singleton batches of $\G$ are identified with their original vertices; these are the \emph{non-heavy} vertices, denoted by $NH$.
We define a function $\hg$ mapping each heavy vertex $v_S$ to $S$, and each non-heavy vertex $v$ to $\{v\}$. 
Intuitively, $\hg$ returns the set of original vertices represented by a vertex in $\hG$.

We establish a key property of $\hG$ in the following lemma.

\begin{restatable}[\app]{lemma}{additiveone}
    \label{lemma:additiveone}
    If the graph $G$ is \dslfree, $\hat{G}$ is a \deslfree graph.
\end{restatable}
\ifthenelse{\equal{\version}{compressed}}
{
}
{
    \begin{proof}
        Consider a chain $\mathcal{C}$ of length $l$ in $\intmath(\hat{G})$. By deleting all heavy vertices from each set in the chain, we show that the resulting sets belong to $\intmath(G)$, thereby proving the existence of a chain in $\intmath(G)$ of length $l-1$.

        To establish this, we examine the structural properties of the intersection closure. Any set in the closure can be viewed as the intersection of closed neighborhoods of a specific vertex subset. We define the functions
        $g_G : 2^{V(G)} \rightarrow 2^{V(G)}$ and $g_{\hat{G}} : 2^{V(\hat{G})} \rightarrow 2^{V(\hat{G})}$
        on the subsets of vertices for $G$ and $\hat{G}$, respectively. For a set $S$, $g(S)$ denotes the set of vertices whose closed neighborhoods intersect to yield exactly $S$. Under this formulation, we can verify the claim by leveraging the property that the heavy vertices form an independent set in $g_{\hat{G}}$.

        \begin{claim}
            \label{claim:containing-closures}
            For any set $S \in \intmath(\hat{G})$, the intersection $S \cap V(G)$ is an element of $\intmath(G)$.
        \end{claim}

        \begin{claimproof}
            Suppose $S = N_{\hG}[X]$ for some vertex set $X \subseteq V(\hat{G})$. We observe that $S \cap V(G) = N_G[\hg(X)]$, where $\hg(X)$ is the set of vertices in $G$ obtained by replacing each heavy vertex in $X$ with its corresponding vertex set in $G$. Essentially, the neighborhood closure in $G$ of this expanded set $\hg(X)$ yields exactly $S \cap V(G)$. Since $S \cap V(G)$ is expressed as an intersection of closed neighborhoods in $G$, it follows that $S \cap V(G) \in \intmath(G)$.
        \end{claimproof}
        Let us examine two sets $S_1$ and $S_2$ in the chain $C$ of $\intmath(\hat{G})$ such that $S_1 \supset S_2$. We say $u$ is in the heavy vertex $s$ if $\hg(s)$ contains $u$.
        \begin{claim}
            If $S_1$ contains more than one heavy vertex, then a heavy vertex $s$ is an element of $S_1$ if and only if $\hg(s) \subseteq S_1$.
        \end{claim}

        \begin{claimproof}
            Since $S_1$ contains at least two heavy vertices, we first observe that $g_{\hat{G}}(S_1)$ cannot contain any heavy vertices. This is because heavy vertices form an independent set in $\hat{G}$, and a heavy vertex $s$ is only contained in its own closed neighborhood. Therefore, we must have $g_{\hat{G}}(S_1) \subseteq V(G)$.

            Recall that $S_1$ is defined as the intersection of the closed neighborhoods of vertices in $g_{\hat{G}}(S_1)$. A heavy vertex $s$ belongs to $S_1$ if and only if $s$ is adjacent to every vertex in $g_{\hat{G}}(S_1)$. By the construction of $\hat{G}$, a non-heavy vertex $v$ is adjacent to a heavy vertex $s$ if and only if $v$ is adjacent to every vertex in the set $\hg(s)$. Consequently, $s \in S_1$ if and only if every $v \in g_{\hat{G}}(S_1)$ is adjacent to every $u \in \hg(s)$, which is equivalent to the condition $\hg(s) \subseteq S_1$.
        \end{claimproof}
        The above claim establishes that if a heavy vertex is not present in $S_2$ but is present in $S_1$, then at least one of its vertices must also not be present in $S_2$ but was present in $S_1$. Therefore, $S_1 \cap V(G)$ and $S_2 \cap V(G)$ are two distinct sets in $\intmath(G)$ provided they contain more than one heavy vertex.

        Now let us assume we are in the case where there are no heavy vertices in $S_1$ and $S_2$. Then, by \Cref{claim:containing-closures}, $S_1$ and $S_2$ are two distinct sets in $\intmath(G)$.
        Now, we need to care when the number of heavy vertices in such a set $S_1$ is exactly one. Now the chain length can only decrease by one only when $S_1 \cap V(G)$ and $S_2 \cap V(G)$ are identical when $S_1$ contains exactly one heavy vertex and $S_2$ does not contain any heavy vertex.

    \end{proof}
}

We now define the \cds optimization problem.

\begin{definition}
    \label{def:cdhs}
    $CDS(G,k,S) = \begin{cases}
            \infty & \text{if } S \text{ is not a connected dominating set for } G \\
            k+1    & \text{if } |S|>k                                              \\
            |S|    & \text{otherwise}                                              \\
        \end{cases}$
\end{definition}
\ifthenelse{\equal{\version}{compressed}}
{
}
{

The following propositions will help us in our kernel.
\begin{definition}[\cite{DBLP:journals/siamdm/EibenKMPS19}]
    Let $D$ be a connected graph and $t\in {\mathbb N}$.  A
    \emph{$(D,t)$-covering family} is a family $\mathcal{F}(D,t)$ of
    connected subgraphs of $D$ such that $(i)$ for each
    $T \in {\mathcal{F}(D,t)}$, $\vert V(T) \vert \leq 2t$ and $(ii)$
    $\bigcup_{T \in \mathcal{F}(D,t)}{V(T)} = V(D)$.
\end{definition}

\begin{proposition}[\cite{DBLP:journals/siamdm/EibenKMPS19}]\label{prop:covering-family}
    Let $D$ be a connected graph and $t\in {\mathbb N}$.  Then there is
    a $(D,t)$-covering family $\mathcal{F}(D,t)$ such that
    $|\mathcal{F}(D,t)| \leq \frac{|V(D)|}{t} + 1$, and
    $\sum_{T \in \mathcal{F}(D,t)}{|V(T)|} \leq (1 + \frac{1}{t}) |V(D)|
        + 1$.
\end{proposition}

\begin{proposition}[\cite{DBLP:journals/siamdm/EibenKMPS19}]
    \label{prop:approx-cds-by-ds}
    Let $G$ be a graph, $X\subseteq V(G)$, such that $G[X]$ is connected,
    and let $D$ be an $X$-dominator such that $G[D]$ has at most $p$
    connected components. Then a set $Q \subseteq X$ of size at most $2p$
    such that $G[D \cup Q]$ is connected, can be computed in polynomial
    time.
\end{proposition}

By leveraging the structural properties of $\hat{G}$ and the preceding propositions, we establish that \cds for \dslfree graphs admits a polynomial $(1+\epsilon)$-lossy compression to \cds for \deslfree graphs.
}


\begin{proposition}[\cite{DBLP:journals/siamdm/EibenKMPS19}]
    \label{prop:approx-cds-by-ds}
    Let $G$ be a graph, $X\subseteq V(G)$, such that $G[X]$ is connected,
    and let $D$ be an $X$-dominator such that $G[D]$ has at most $p$
    connected components. Then a set $Q \subseteq X$ of size at most $2p$
    such that $G[D \cup Q]$ is connected, can be computed in polynomial
    time.
\end{proposition}
The following paragraph details an algorithm that returns a $\textrm{poly}(k)$ sized instance, which we shall prove is a polynomial $(1+\epsilon)$-lossy compression.

\subparagraph*{Compression algorithm.}
Let $t$ be any fixed integer (to be determined later based on $\epsilon$).
For each subset $Q \subseteq \C$ of size at most $2t$, we run the \gst algorithm (\Cref{prop : group steiner tree algo}) in $\hat G$, using the classes in $Q$, which are pairwise-disjoint subsets of $V(G) \subseteq V(\hat{G})$, as the groups, to find a tree $T_Q$. If $|T_Q| \le 2t$, we mark its vertices; otherwise, we discard it. Let $M$ denote the set of all marked vertices. If $M$ is not a dominating set of $G$, we return a trivial \no-instance with parameter $k=0$.

After performing the above polynomial-time check, we make $M$ connected by adding at most $2|M|$ vertices from $G$, using \Cref{prop:approx-cds-by-ds}.
We denote the set of connecting vertices by $C_{\rm conn}$. By construction, $G[M \cup C_{\rm conn}]$ is connected.
We return the instance $I'=(G', k)$, where $G'$ is the induced subgraph  $\hG[M\cup C_{\rm conn} \cup H \cup NH]$. 


By \Cref{lemma:additiveone} and the fact that semi-ladder index is closed under taking induced subgraphs, $G'$ is $(d+2)$-semi-ladder-free.

\begin{restatable}[\app]{lemma}{solnlifting}
    \label{lemma:soln-lifting}
    Let $S'$ denote a connected dominating set of $G'$. Then, there exists a polynomial-time algorithm to output a connected dominating set of $G$ of size at most $|S'|$.
\end{restatable}
\ifthenelse{\equal{\version}{compressed}}
{
}
{
    \begin{proof}
        We repeatedly remove heavy vertices from the solution.  Let $h=v_B\in S'\cap H$ be a heavy vertex.  If $|S'|=1$, then the existence of the connected dominating set $\{h\}$ implies that some vertex of $G$ dominates the batch $B$; replacing $h$ by such a vertex gives a connected dominating set of $G$ of size one.  Assume therefore that $|S'|>1$.

        Since $G'[S']$ is connected, $h$ has a neighbor $u\in S'\setminus\{h\}$.  By the construction of $\hat G$, every non-heavy neighbor of $h$ is adjacent in $G$ to every vertex of the batch $B$.  Choose any $b\in B$ and set $S''=(S'\setminus\{h\})\cup\{b\}$.  The vertex $b$ is adjacent to every selected neighbor of $h$, so replacing $h$ by $b$ preserves connectivity.  Moreover, every original vertex whose domination was certified through $h$ is dominated by $b$ because $h$ represents the batch $B$.  Thus $S''$ is again a connected dominating set of $G'$ with one fewer heavy vertex and $|S''|\le |S'|$.  Iterating this replacement yields a connected dominating set contained in $V(G)$; by the preservation property of the $k$-grouped domination core, it dominates $G$.
    \end{proof}
}

In the following lemma, we analyze the size and the approximation guarantee of the compression algorithm.

\begin{restatable}[\app]{lemma}{lossycompsizeapprox}
\label{lem:lossy-comp-size-approx}
The size of the instance $I'$ is a polynomial in $k$ and ${\rm OPT}(I',k) \le (1+\epsilon) {\rm OPT}(I,k)$.
\end{restatable}
\ifthenelse{\equal{\version}{compressed}}
{
}
{
    \begin{proof}
        The size bound follows from the core bounds.  We have $|H\cup NH|\le |\G|\le k^d$.  The algorithm marks only trees for subsets of the $k$-dominator core of size at most $2t$, and each such tree has at most $2t$ vertices.  Since $|\C|=\Oh(k^{d^2})$, we get $|M|\le |\C|^{2t}\cdot 2t=k^{\Oh(td^2)}$, and $|M\cup C_{\rm conn}|\le 3|M|$.  Hence $|V(G')|\le k^d+k^{\Oh(td^2)}$.  With $t=\Theta(1/\epsilon)$ this is $k^{\Oh(d^2/\epsilon)}$.

        We next prove the approximation guarantee.  Let $L=\lceil 8/\epsilon\rceil$.  As a preprocessing step, the compression algorithm first checks by brute force whether $G$ has a connected dominating set of size at most $L$.  If it does, it stores such an optimum small solution and outputs a constant-size equivalent instance whose lifting algorithm returns it.  Thus assume below that $\operatorname{OPT}(I,k)>L$.

        Let $D^\star$ be an optimum connected dominating set of $G$ with $|D^\star|\le k$; if no such set exists, then the claim is immediate from the definition of the thresholded objective.  Let $D_0\subseteq D^\star$ be an inclusion-wise minimal dominating set.  Consider a $(D^\star,t)$-covering family as in \Cref{prop:covering-family}.  For each connected subgraph $T$ in this family, let $Q_T$ be the set of $k$-dominator core classes that contain vertices of $T\cap D_0$.  The unique-contribution property of the $k$-dominator core implies $|Q_T|\le |V(T)|\le 2t$, so the compression algorithm marks a tree of size at most $2t$ hitting all groups in $Q_T$.  Taking the union of these marked trees gives a set that dominates all batches of the $k$-grouped domination core, and hence dominates $G$.

        The covering-family bound gives total marked size at most $(1+1/t)|D^\star|+1$.  These marked vertices induce at most $|D^\star|/t+1$ connected components, so \Cref{prop:approx-cds-by-ds} adds at most $2|D^\star|/t+2$ connector vertices.  Therefore
        \[
        \operatorname{OPT}(I',k)\le (1+3/t)|D^\star|+3.
        \]
        Choose $t\ge 8/\epsilon$.  Since $|D^\star|>8/\epsilon$, the additive term satisfies $3\le (\epsilon/2)|D^\star|$, and $3|D^\star|/t\le (3\epsilon/8)|D^\star|$.  Hence $\operatorname{OPT}(I',k)\le (1+\epsilon)|D^\star|=(1+\epsilon)\operatorname{OPT}(I,k)$.
    \end{proof}
}


By \Cref{lemma:soln-lifting}, we have $CDS(G,k,S) \le CDS(G',k,S)$.
By \Cref{lem:lossy-comp-size-approx}, we have ${\rm OPT}(I',k) \le (1+\varepsilon){\rm OPT}(I,k)$.
Therefore, $$\frac{CDS(G,k,S)}{{\rm OPT}(I,k)} \le \frac{CDS(G',k,S)}{\frac{1}{1+\epsilon}{\rm OPT}(I',k)} \le (1+\epsilon) \frac{CDS(G',k,S)}{{\rm OPT}(I',k)}.$$
Combining this with the solution lifting algorithm in \Cref{lemma:soln-lifting} and the above approximation guarantee, we obtain \Cref{thm: lossy_compression}.

\section{Acknowledgements}
\noindent \textbf{AI Disclosure.}
The authors declare that they have not used AI-assisted technologies in creating this article. OpenAI's ChatGPT was used to improve the English language in the article.

\bibliography{bibs}

\appendix

\section{Appendix}
\subsection{Basics of Parameterized Complexity}
The following definitions are standard and can be found in~\cite{DBLP:books/sp/CyganFKLMPPS15}.
A parameterized problem (for e.g., \ds parameterized by solution size) is said to be \textit{fixed-parameter tractable} (\fpt or $\fpt$(parameter), in short) if an instance $(\mathcal{I},k)$ of the problem can be solved in time $f(k) \cdot |\mathcal{I}|^{\Oh(1)}$ for some computable function $f$.

\begin{definition}\label{def:paraOptProblem}
A parameterized optimization (minimization or maximization) problem~$\Pi$ is a computable function
\[
\Pi : \Sigma^* \times \mathbb{N} \times \Sigma^* \rightarrow \mathbb{R} \cup \{\pm\infty\}.
\]

An \emph{instance} of~$\Pi$ is a pair~$(I,k) \in \Sigma^*\times \mathbb{N}$,  
and a \emph{solution} to~$(I,k)$ is a string~$s \in \Sigma^*$ with~$|s| \le |I|+k$.  
The \emph{value} of~$s$ is~$\Pi(I,k,s)$.
\end{definition}

\begin{definition}\label{def:paraOpt}
For a parameterized minimization problem~$\Pi$, the \emph{optimum value} of~$(I,k)$ is
\[
OPT_\Pi(I,k) = \min_{\substack{s \in \Sigma^* \\ |s| \le |I|+k}} \Pi(I,k,s).
\]
For a parameterized maximization problem~$\Pi$, it is
\[
OPT_\Pi(I,k) = \max_{\substack{s \in \Sigma^* \\ |s| \le |I|+k}} \Pi(I,k,s).
\]
A solution~$s$ is \emph{optimal} if~$\Pi(I,k,s) = OPT_\Pi(I,k)$.
\end{definition}


\section{Missing proofs and details from \Cref{sec:prelims}}
\label{sec:app-prelims}
    
\semiladderfreegraphtwo*
    \begin{proof}
    Suppose, for contradiction, that $G$ contains a $3d$-semi-ladder.
    Then there exist vertices and neighborhoods
    \[
    a_1,\dots,a_{3d}
    \quad\text{and}\quad
    N_G[b_1],\dots,N_G[b_{3d}]
    \]
    such that
    \[
    a_i \notin N_G[b_i]
    \quad\text{and}\quad
    a_i \in N_G[b_j] \text{ for all } i<j.
    \]
    
    Note that if the vertices $a_1,\dots,a_{3d}$ and $b_1,\dots,b_{3d}$ were all distinct, we would have found a $3d$-semi-ladder as an induced subgraph in $G$.
    This would contradict the assumption that $G$ is semi-induced $d$-semi-ladder-free, since the sets $N_G[b_i]$ can be correspond to $b_i$ themselves while maintaining the neighborhood relation.
    
    Now, consider the following greedy process. We start with $i=1$. By definition of a semi-ladder, $a_i \notin N_G[b_i]$ but is adjacent to $N_G[b_j]$ for all $j>i$. We add $a_i$ and $b_i$ to the sets $L$ and $R$, respectively. Note that $a_i \neq b_i$ as $a_i \notin N_G[b_i]$. We then follow a row marking process. Given $a_i$ and $b_i$, we mark the following rows:
    \begin{itemize}
        \item Row $j$ where $a_j=b_i$,
        \item Row $j$ where $b_j=a_i$.
    \end{itemize}
    Hence, each time we add a vertex to $L$ and $R$, we mark at most two rows. We find a row $i'$ that is not marked and add $a_{i'}$ and $b_{i'}$ to $L$ and $R$, respectively. We repeat this process until we have added $d$ vertices to $L$ and $R$. We can always find such an unmarked row as we start with $3d$ rows and each time we add a vertex to $L$ and $R$, we mark at most two rows.
    
    After this process, we have two sets $L$ and $R$ of size $d$ each such that for all $i,j \in [d]$ with $i < j$, we have $(a_i,b_j) \in E(G)$, and for all $i \in [d]$, we have $(a_i,b_i) \notin E(G)$. This contradicts the assumption that $G$ is semi-induced $d$-semi-ladder-free.
    \end{proof}
\section{Missing proofs and details from \Cref{sec:csc-dslfree}}
\label{sec:app-csc-dslfree}

\subparagraph{Construction.}
Let $k,m\in\mathbb{N}$. Define the universe
$
\mathcal U \;:=\; \{t\} \cup \{u_i \mid i\in [k]\} \cup \{b_{i,j}\mid i\in[k],\, j\in[m]\}.
$
Let
$
B := \{t\}\cup\{b_{i,j}\mid i\in[k],\, j\in[m]\},
$
and for each $i\in[k], j\in[m]$,
let
$
S_{i,j} := \{u_i,\, b_{i,j}\}
.
$
Define the family $\mathcal F:=\{B\}\cup\{S_{i,j}:i\in[k],j\in[m]\}$.
All sets in $\mathcal F$ are pairwise distinct (each $S_{i,j}$ contains a unique element $b_{i,j}$, and $B$ contains $t$).
An illustration of the incidence graph of $(\mathcal U,\mathcal F)$ is given in Figure~\ref{fig:planar-incidence}.

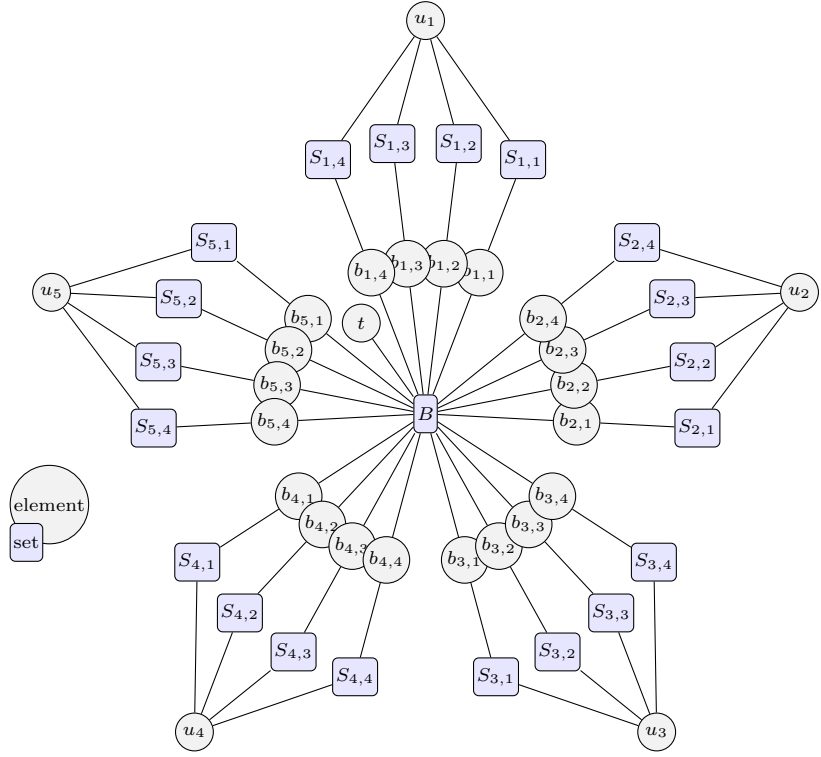
\begin{figure}[t]
  \centering
  \begin{tikzpicture}[
    font=\scriptsize,
    element/.style={circle, draw, fill=gray!10, inner sep=1pt, minimum size=5mm},
    setv/.style={rectangle, rounded corners=2pt, draw, fill=blue!10, inner sep=1.2pt, minimum height=5mm},
    edge/.style={draw}
  ]
  \def\k{5}         
  \def\m{4}         
  \def\Ru{5.2}      
  \def\Rs{3.6}      
  \def\Rb{2.0}      

  \node[setv] (B) at (0,0) {$B$};
  \node[element] (t) at (-0.85,1.2) {$t$};
  \draw[edge] (t) -- (B);

  \pgfmathsetmacro{\angstep}{360/\k}
  \pgfmathsetmacro{\spread}{min(14, 0.70*\angstep/max(1,\m-1))}

  \foreach \i in {1,...,\k}{
    \pgfmathsetmacro{\theta}{90-(\i-1)*\angstep}
    \node[element] (u\i) at (\theta:\Ru) {$u_{\i}$};

    \foreach \j in {1,...,\m}{
      \pgfmathsetmacro{\delta}{(\j-(\m+1)/2)*\spread}

      \node[setv]    (S\i\j) at ({\theta+\delta}:\Rs) {$S_{\i,\j}$};
      \node[element] (b\i\j) at ({\theta+\delta}:\Rb) {$b_{\i,\j}$};

      \draw[edge] (u\i) -- (S\i\j);
      \draw[edge] (S\i\j) -- (b\i\j);
      \draw[edge] (b\i\j) -- (B);
    }
  }

  \node[element, anchor=west] at (-5.5,-1.2) {element};
  \node[setv,    anchor=west] at (-5.5,-1.7) {set};
  \end{tikzpicture}
  \caption{A planar drawing of the incidence graph. 
  Element-vertices are circles, and set-vertices are boxes. 
  For each $i$, the paths $u_i $-$ S_{i,j} $-$ b_{i,j} $-$ B$ are routed using disjoint wedges, ensuring planarity.
  }
  \label{fig:planar-incidence}
\end{figure}
\planarmanyminimal*
\begin{proof}
Applying the following two results together, on the construction, gives us the desired result.
\begin{lemma}\label{lem:many-minimal}
The instance $(\mathcal U,\mathcal F)$ has exactly $m^{k}$ inclusion-wise minimal set covers.
\end{lemma}
\begin{proof}
Any set cover $\mathcal C\subseteq\mathcal F$ must contain $B$, since $t$ appears only in $B$.
Once $B$ is chosen, every $b_{i,j}$ is covered, so the only remaining elements are $u_1,\dots,u_k$.
For each $i\in[k]$, the element $u_i$ appears only in $S_{i,1},\dots,S_{i,m}$, hence $\mathcal C$ must contain at least
one set from this block.

If $\mathcal C$ is inclusion-wise minimal, then for each $i$ it contains \emph{exactly one} set from
$\{S_{i,1},\dots,S_{i,m}\}$: if it contained two, one would be redundant because $B$ already covers all $b_{i,j}$'s.
Thus every inclusion-wise minimal set cover has the form
$\mathcal C \;=\; \{B\}\ \cup\ \{S_{1,j_1},S_{2,j_2},\dots,S_{k,j_k}\}$ with $j_i\in[m]\ \text{for all } i\in[k]$. 
Conversely, any such choice covers $\mathcal U$ and is inclusion-wise minimal (removing $B$ uncovers $t$, and removing
$S_{i,j_i}$ uncovers $u_i$). Hence there are exactly $m^{k}$ such covers.
\end{proof}

\begin{lemma}\label{lem:planar-k22}
The incidence graph of $(\mathcal U,\mathcal F)$ is planar and $K_{2,2}$-free (and hence $K_{d,d}$-free for every $d\ge 2$).
\end{lemma}
\begin{proof}
For planarity, place the set-vertex $B$ in the center and the element-vertices $u_1,\dots,u_k$ on a circle around $B$.
For each $i\in[k]$, reserve a disjoint wedge region between $u_i$ and $B$, and inside it draw $m$ internally disjoint
paths $u_i-S_{i,j}-b_{i,j}-B$ (for $j\in[m]$). Place $t$ near $B$ and draw the edge $tB$. No edges cross, so the
incidence graph is planar. See Figure~\ref{fig:planar-incidence} for an illustration. 

For $K_{2,2}$-freeness, it suffices to show that any two distinct sets in $\mathcal F$ intersect in at most one element.
Indeed, for $S_{i,j}$ and $S_{i',j'}$, the intersection is empty if $i\neq i'$ and equals $\{u_i\}$ if $i=i'$ and
$j\neq j'$. Moreover, $B\cap S_{i,j}=\{b_{i,j}\}$. Hence no two sets share two elements, so the incidence graph has no
$K_{2,2}$.
\end{proof}
This concludes the proof of the theorem.
\end{proof}
\findcompactrepresentations*
\begin{proof}
We will use the following lemma to prove the theorem.
\lcompactrepresentationcorrectness*
\begin{proof}
    {
We now present our first key lemma, which links latent covers to the base case of our algorithm. It shows that once an accumulated sequence from the intersection closure covers the universe, the \textsc{Solve} procedure can efficiently construct a valid tuple representing all its extensions.



\begin{restatable}{claim}{partialcoverfound}
    \label{lemma : partial_cover_found}
    Consider a set system $(\U, \F)$.
  Suppose that $(I_1,\ldots,I_{\ell})$ is a latent cover where $I_i \in \intmath(\F)$ for each $i \in [\ell]$. Then, procedure \solve\callfamily{I} outputs an $\ell$-tuple $\mathbb{T} = (F_1,\ldots,F_\ell)$ such that
    \begin{itemize}
        \item for each sequence $(S_1,\ldots,S_\ell) \in F_1 \times \ldots \times F_\ell$, we have $\bigcup_{i=1}^\ell S_i = \mathcal{U}$ and the families $\{F_i : ~i\in [\ell]\}$ are pairwise disjoint.
        \item The tuple $\mathbb{T}$ \contains every sequence $(S_1, \dots, S_\ell)$ such that $(S_1, \dots, S_\ell)$ (1) extends $(I_1, \dots, I_\ell)$; and (2) forms a minimal set cover.
        \item The running time of \solve\callfamily{I} is $\Oh(\ell mn)$.
    \end{itemize}
\end{restatable}

    \begin{claimproof}
    When Line \ref{step: check}
    in \Cref{alg:find_compact_representations} is triggered, we have
    $F_i = \{S \mid S \supseteq I_i \text{ and } S \in \F\}$ for each $i \in [\ell]$.
    Hence, all minimal set covers that minimally extend $(I_1,\ldots,I_{\ell})$ are in $(F_1,\ldots,F_{\ell})$.
    Moreover, for each $(f_1,\ldots,f_\ell) \in F_1\times\cdots\times F_\ell$, it holds that $\bigcup_{i=1}^\ell f_i = \mathcal{U}$ since $(I_1,\ldots,I_{\ell})$ is a latent cover.

    {Let $S^{\star}$ be a minimal set cover of size $\ell$ that minimally extends $(I_1, \ldots, I_{\ell})$, with $S^{\star}_i \supseteq I_i$ for each $i\in [\ell]$. 
    Suppose, for contradiction, that $S^{\star}$ has a set $S \in \mathcal{F}$ such that $S \in F_j \cap F_k$ for distinct indices $j, k \in [\ell]$. 
    This implies $S \supseteq I_j$ and $S \supseteq I_k$.
    Without loss of generality, assume that $S^{\star}_j = S$ (note that due to minimality of $S^\star$, this implies that $S^{\star}_k \ne S$).
    Consider the subfamily $S' = S^{\star} \setminus \{S^{\star}_k\}$. 
    Since $\bigcup_{i=1}^\ell I_i = \mathcal{U}$ and every element in $S'$ is a superset of its corresponding $I$, the union of sets in $S'$ covers $\bigcup_{i=1}^\ell I_i = \mathcal{U}$ (since $S^{\star}_j = S \supseteq I_k$).
    Thus, $S^{\star} \setminus \{S^{\star}_k\}$ is a valid set cover of size $\ell-1$, which contradicts the assumption that $S^{\star}$ is an inclusion-wise minimal set cover of size $\ell$.
    Therefore, no minimal set cover of size $\ell$ can use a set that belongs to multiple $F_i$'s, and such sets can be safely removed to ensure the families are disjoint.}
    
  \textit{Running time analysis.} Every set $S\subseteq \U$ can be represented by an $n$-bit string where the $e$-th bit indicates whether an element $e$ belongs to the set $S$ or not. 
    This takes $\Oh(n)$ time per set.
    Now for each $i \in [\ell], S' \in \F$, we calculate bit strings for $I_i$ and $S'$ and include $S'$ in $F_i$ if all the elements that are present in $I_i$ are also present in $S'$. Now similarly any set family $F \subseteq \F$ can be represented by an $m$-bit string. For each $i \in [\ell]$, we compare $\ell-1$ many bit string and remove sets from $F_i$ that are present in any one of $\{F_1,\ldots,F_\ell\} \setminus \{F_i\}$.
    All of these operations take $\Oh(\ell mn)$ time.
\end{claimproof}
}
    We prove the statement by induction on the value of $\mu\fami{I}$.

    \noindent\textbf{Base Case:} $\mu\fami{I} = 0$.
    We first show that when $\mu\fami{I}=0$, either the tuple \family{I} already constitutes a latent cover, or no minimal set cover of size exactly $\ell$ exists that minimally extends \family{I}.

    If \family{I} is a latent cover, then by \Cref{lemma : partial_cover_found}, \solve\callfamily{I} correctly outputs a tuple \containing all minimal set covers that minimally extend \family{I} in time $\Oh(\ell mn)$.

    Now, suppose that \family{I} is not a latent cover. Since $I_i \in \intmath(\F)$ for all $i \in [\ell]$, by \Cref{prop : bounded_length_of_int_closure}, we know that $l(I_i) \le d$ for all $i \in [\ell]$. Given that $\mu\fami{I} = 0$, it follows that $l(I_i) = d$ for all $i \in [\ell]$. Suppose, for contradiction, that there exists a minimal set cover \family{S} that minimally extends \family{I}. Then, there must exist an element $e \in \U \setminus (I_1 \cup \ldots \cup I_{\ell})$. Since \family{S} covers $\U$, $e \in S_j$ for some $j$. This implies that $l(S_j) = d+1$ (since $S_j \supsetneq I_j$), contradicting the bound given in \Cref{prop : bounded_length_of_int_closure}. Therefore, no such \family{S} exists in this case, completing the base step.

    \noindent
    \textbf{Induction Hypothesis:}
    Assume that the lemma holds for all instances of \family{I} such that $\mu\fami{I} < t$. That is, for any tuple with a smaller measure, the algorithm correctly outputs all required $\ell$-tuples with the claimed properties and within the stated time bound.

    \smallskip
    \noindent
    \textbf{Inductive Step:}
    We now consider the case $\mu\fami{I} = t > 0$. There are two possibilities to analyze depending on whether \family{I} forms a latent cover or not.

    \medskip
    \noindent
    \textit{Case 1: \family{I} is a latent cover.}
    In this scenario, by \Cref{lemma : partial_cover_found}, \solve\callfamily{I} outputs a tuple satisfying all the desired properties, including the fact that it \contains~all minimal set covers that minimally extend \family{I}, and it does so in time $\Oh(\ell mn)$.

    \medskip
    \noindent
    \textit{Case 2: \family{I} is not a latent cover.}
    Let \family{S} be a minimal set cover that minimally extends \family{I}. Since \family{I} is not a latent cover, there must exist an element $e \in S_j$ for some $j$ that is uncovered by \family{I}, i.e., $e \in \U \setminus (I_1 \cup \ldots \cup I_{\ell})$. \Cref{alg:find_compact_representations} explicitly branches on which set $S_j$ covers the uncovered element $e$ in Line 8.

   Since \family{S} minimally extends \family{I}, we have $I_j \subseteq S_j$. Moreover, since $e \in S_j$, it follows that $I_j \cup \{e\} \subseteq S_j$. We define 
    \[
        I_{je} = \bigcap_{I_j\cup\{e\}\subseteq S \in F} S,
    \]
    that is, $I_{je}$ is the intersection of all sets in $\F$ that contain $I_j \cup \{e\}$. Let $F_{je}$ denote the family of all sets $S \in \F$ such that $I_j \cup \{e\} \subseteq S $.

    Clearly, $F_{je} \neq \emptyset$ since $S_j \in F_{je}$. 
    Moreover, $I_{je} \in \intmath(\F)$ and, by construction, $l(I_j) < l(I_{je})$ because $I_j \subsetneq I_{je}$ (since $e \in I_{je}\setminus I_j$).
    Recall that the measure is $\mu\fami{I} = \ell d - \sum^{\ell}_{i=1} l(I_i)$. Since the chain length $l(\cdot)$ strictly increases for $I_{je}$, the sum subtracted from $\ell d$ strictly increases. Hence, the overall measure strictly decreases: $\mu(I_1, \ldots, I_{je}, \ldots, I_{\ell}) < \mu\fami{I} = t$.
    This allows us to apply the induction hypothesis on the instance $(I_1, \ldots, I_{je}, \ldots, I_{\ell})$.

    By the induction hypothesis, there exists an $\ell$-tuple $\mathbb{T}$ in the family of $\ell$-tuples returned by \solve$(I_1,\ldots,I_{je},\ldots,I_{\ell},\U,\F)$ such that $\mathbb{T}$ \contains~\family{S}. Consequently, due to Line 14 of the algorithm, there must also exist a corresponding $\ell$-tuple $\mathbb{T}$ in the family of $\ell$-tuples returned by \solve$(I_1,\ldots,I_{j},\ldots,I_{\ell},\U,\F)$ that \contains ~\family{S} as desired.

    Note that we are in the case where \family{I} is not a latent cover. For each $j \in [\ell]$ and each $I_{je}$ constructed from a non-empty $F_{je}$, we have $l(I_{je}) > l(I_j)$, implying that $\mu(I_1, \ldots, I_{je}, \ldots, I_{\ell}) < t$. Thus, we may again apply the induction hypothesis to all such recursive instances.

    \begin{sloppypar}
        By the induction hypothesis, for every $\ell$-tuple $\mathbb{T}=(F_1,F_2,\ldots,F_{\ell})$ returned by \solve$(I_1,\ldots,I_{je},\ldots,I_{\ell},\U,\F)$, and for every sequence $(X_1,\ldots,X_\ell) \in F_1 \times \cdots \times F_\ell$, we have $\bigcup_{i=1}^{\ell} X_i = \mathcal{U}$. 
        Since in Line 14, the algorithm takes the union over all such recursive calls, it follows that in the final output of \solve\callfamily{I}, for every $\ell$-tuple $\mathbb{T} = (F_1,F_2,\ldots,F_{\ell})$ and every $(X_1,\ldots,X_\ell) \in F_1\times\cdots\times F_\ell$, the union $\bigcup_{i=1}^\ell X_i = \mathcal{U}$ also holds.

    \end{sloppypar}


{\it Bounding the number of $\ell$-tuples.}
  By the induction hypothesis, the number of $\ell$-tuples returned by each call to \solve$(I_1,\ldots,I_{je},\ldots,I_{\ell},\U,\F)$ is at most $\ell^{\mu(I_1,\ldots,I_{je},\ldots,I_{\ell})} \le \ell^{t-1}$, for each $j \in [\ell]$ with $F_{je} \neq \emptyset$. Since there are at most $\ell$ such branches, the total number of $\ell$-tuples returned by \solve\callfamily{I} is at most $\ell^t$.

  \begin{sloppypar}
 {\it Running-time analysis.} Each recursive instance of the algorithm \solve$(I_1,\ldots,I_{je},\ldots,I_{\ell},\U,\F)$ (for each valid $j$) takes time at most $\ell^{(t-1)+1} \cdot \Oh(nm)$, and each such instance outputs at most $\ell^{t-1}$ tuples, each of size at most $\ell \cdot n$. Therefore, the total running time across all recursive branches is bounded by $\ell^{t-1} \cdot \ell \cdot \ell \cdot \Oh(nm) = \ell^{t+1} \cdot \Oh(nm)$,
    which establishes the claimed time complexity.
\end{sloppypar}

    This completes the inductive argument, and hence the proof of the lemma. 
\end{proof}

        For each $i \in [k]$, we invoke \solve$(X_1,X_2,\ldots,X_i,\U,\F)$ where $X_j=\bot=\emptyset$ for every $j\in[i]$, and $\ell=i$. Every minimal set cover of size $i$ extends this bottom tuple. Hence, procedure \solve$(X_1,X_2,\ldots,X_i,\U,\F)$ returns a family of tuples $\mathscr{D}_i$ where $|\mathscr{D}_i| \le i^{id}$ that compactly represents all inclusion-wise minimal set covers of size exactly $i$.
    
        We compute $\mathscr{D}$ by taking the union of all $\mathscr{D}_i$ for $i \in [k]$ (that is, $\X{D} = \bigcup_{i\in [k]}\X{D}_i$).
        It compactly represents all inclusion-wise minimal set covers of size at most $k$ in $(\mathcal{U},\mathcal{F})$, satisfying the desired properties in \Cref{theorem:find_compact_representations}.
        Therefore, $|\D|\le k\cdot k^{kd} \le k^{kd+1}$. As every tuple output by the algorithm for all values of $i \in [k]$ satisfies the covering condition, the union also does.
        Moreover, for each $i \in [k]$, computing $\mathscr{D}_i$ requires at most $i^{id+1} \cdot \Oh(nm)$ time. Therefore, the total time to compute $\mathscr{D}$ is bounded by $k^{kd+2} \cdot \Oh(nm)$.
    \end{proof}
\imposscountcovers*
{
    Since there are $k$ choices for each of the $c$ columns, the total number of sets is $|\mathcal{F}| = k^c$. 
    Observe that since every set is of size exactly $c$, every minimal set cover of size at most $k$ is actually a minimum set cover of size exactly $k$.
    Moreover, every minimum set cover is also a partition of the universe: that is, each element of the universe is contained in exactly one $c$-sized set in the set cover.
    Thus, we have the following.
    
    \begin{claim}
        \label{clm:imposs-singletons-in-D}
        In any compact representation of all minimal set covers of size at most $k$ in $(\mathcal{U}, \mathcal{F})$, every sub-family of sets present in any $k$-tuple must be a singleton.
        That is, for any $k$-tuple $\mathbb{T} = (F_1, \dots, F_k)$ in the representation, $|F_i| = 1$ for all $i \in [k]$.
    \end{claim}
    \begin{claimproof}
        By the covering property of a compact representation, every combination of sets $(S_1, \dots, S_k) \in F_1 \times \dots \times F_k$ must cover $\mathcal{U}$. 
    
        Fix arbitrary sets $S_2 \in F_2, \dots, S_k \in F_k$. Since each set covers exactly $c$ elements, their union covers at most $(k-1)c$ elements. 
        This leaves at least $c$ elements of $\mathcal{U}$ uncovered (specifically, exactly one element in each column). 
    
        For the tuple to be valid, \emph{every} set $S_1 \in F_1$ must cover exactly these remaining $c$ elements. 
        However, by the definition of $\mathcal{F}$, there is exactly one set in $\mathcal{F}$ that contains those specific $c$ elements. 
        Therefore, $F_1$ is a singleton. 
        By symmetry, this holds for all $F_i$.
    \end{claimproof}
    
    Let $\mathscr{D}$ denote \emph{any} compact representation of all minimal set covers of size at most $k$ in $(\mathcal{U}, \mathcal{F})$.
    Consider any minimal set cover $(S_1, \dots, S_k)$ in the set system.
    Observe that such a set cover always exists; for example, one can choose $S_i$ to be exactly the elements in the $i$-th row of the grid.
    By definition, there is some $k$-tuple in $\mathscr{D}$ that \contains this set cover.
    Moreover, by \Cref{clm:imposs-singletons-in-D}, this is the \emph{unique} set cover that this tuple \contains.
    Thus, the number of $k$-tuples in $\mathscr{D}$ is tightly lower-bounded by the number of minimal set covers of size $k$ in $(\mathcal{U}, \mathcal{F})$.
    
    \begin{claim}
        $|\mathscr{D}| \ge \text{(number of minimal set covers of size $k$)}$.
    \end{claim}
    
    To count the number of minimal set covers of size $k$, we count the number of ways to partition the grid into $k$ disjoint sets of size $c$. 
    We process the universe column by column:
    \begin{itemize}
        \item Column 1: There are $k$ elements in the first column to be distributed among $k$ \emph{unlabelled} sets. 
        There is exactly $1$ way to initialize this partition.
        \item Columns 2 to $c$: For each subsequent column, we must assign its $k$ elements to the $k$ established sets. 
        There are exactly $k!$ ways to distinctly map the elements of the column to the $k$ sets.
    \end{itemize}
    Since the choices for columns $2$ to $c$ are completely independent, we have the following.
    \begin{claim}
        \label{clm:imposs-count-covers}
        The number of minimal set covers of size $k$ in $(\mathcal{U}, \mathcal{F})$ is $(k!)^{c-1}$.
        Consequently, $|\mathscr{D}| \ge (k!)^{c-1}$.
    \end{claim}
}
\tightnessdslfree*
\begin{proof}
    When $c=d-2$, all sets in $\mathcal{F}$ are of size $d-2$.
    By definition, a $d$-semi-ladder requires $d$ sets $S_1, \dots, S_d$ and $d$ elements $x_1, \dots, x_d$ such that $x_i \in S_j$ for $j > i$, and $x_i \notin S_i$. 
    Specifically, the set $S_d$ must contain the $d-1$ elements $x_1, \dots, x_{d-1}$. 
    Since no set in $\mathcal{F}$ can contain two elements from the same column, these $d-1$ elements must belong to distinct columns. 
    However, the grid only has $(d-2)$ columns. Therefore, such a set $S_d$ cannot exist.
    Thus, there is no $d$-semi-ladder in $\mathcal{F}$.
    
\end{proof}
\insertlabelhere*
\begin{proof}
    We choose $k$ to be any integer sufficiently large such that $\log_k(k!) > q$.
    Consider our construction $(\mathcal{U}, \mathcal{F})$ with $k$ rows and $c$ columns, where $c$ is an integer we will choose later.
    
    Recall that $|\mathcal{U}| = kc$ and $|\mathcal{F}| = k^c$.
    For any positive $c$, we have $kc \le k^c$, and thus $|\mathcal{U}| + |\mathcal{F}| \le 2k^c$. 
    
    Let $\mathscr{D}$ denote an arbitrary compact representation of all minimal set covers of size at most $k$.
    By \Cref{clm:imposs-count-covers}, we have $|\mathscr{D}| \ge (k!)^{c-1}$.
    Therefore, for the size to exceed the stated \fpt bound, it suffices to choose $c$ such that:
    \[
        |\mathscr{D}| \ge (k!)^{c-1} > f(k) \cdot (2k^c)^q \ge f(k) \cdot (|\mathcal{U}| + |\mathcal{F}|)^q.
    \]
    Taking the logarithm base $k$ yields:
    \[
        (c-1) \log_k(k!) > \log_k(f(k)) + q\log_k(2) + c \cdot q.
    \]
    Rearranging the terms gives:
    \[
        c \cdot \left( \log_k(k!) - q \right) > \log_k(f(k)) + q\log_k(2) + \log_k(k!).
    \]
    Since we chose $k$ to be sufficiently large, the coefficient $(\log_k(k!) - q)$ is a positive constant. 
    Moreover, since $k$ and $q$ are now fixed, the entire right-hand side is a constant.
    Finally, we set $c$ sufficiently large so that the inequality holds, and thus we have the theorem.
\end{proof}
\cscalgo*
\begin{proof}
        Let $\mathscr{D}$ be the family of tuples returned by \Cref{theorem:find_compact_representations} in $k^{kd+2}\cdot \Oh(nm)$ time.
        We would like to invoke the \gst algorithm given by \Cref{prop : group steiner tree algo} to find a solution.
        For each tuple $\mathbb{T}=\{F_1,F_2,\ldots,F_\ell\}$, we define the groups to be $F_i$ for each $i \in [\ell]$. Now we will run the \gst algorithm given by \Cref{prop : group steiner tree algo} on the graph $\hat{G}$ to find a solution. We will call such an instance the \gst instance on tuple $\mathbb{T}$.
        Any solution $T=\{S_1,\ldots,S_\ell\}$ to \gst instance on any tuple is also a solution for the \Csc instance as at least one set of each family must be taken in $T$, which ensures that all elements are covered by condition~\ref{compact_rep:covering_condition}.
        Also, as $T$ is a solution of the \gst instance, $T$ is connected.
        What remains to be shown is that, if the \Csc instance parameterized by solution size is a \yes instance, then there exists a tuple such that \Cref{prop : group steiner tree algo} returns a solution of size at most $k$.
        Let $D^\star$ be an optimal solution of the \Csc instance of size at most $k$. Let $D'\subseteq D^\star$ be a minimal solution. We know by \Cref{def:compact_representations} that there exists a tuple $\mathbb{T}$ in $\mathscr{D}$ that \contains $D'$ as $|D'| \le k$. Since $D^\star$ is a solution for the \gst instance on the tuple $\mathbb{T}$, there exists a solution returned by \Cref{prop : group steiner tree algo} of size at most $k$.
        
        \noindent
        \textit{Running time analysis.} The family of tuples $\mathscr{D}$ of size at most $k^{kd+1}$ can be computed in $k^{kd+2}\cdot \Oh(nm)$ time (\Cref{theorem:find_compact_representations}).
        Now, for each tuple, the \gst algorithm (\Cref{prop : group steiner tree algo}) takes $2^\ell \cdot n^{\Oh(1)}$ time, the entire algorithm runs in $k^{kd+2}\cdot 2^k \cdot n^{\Oh(1)}$ time.
        \Ma{\textit{Space complexity analysis.}
        We note that if we first store all the tuples in $\mathscr{D}$, then the space complexity of the algorithm is $k^{kd+1}\cdot n^{\Oh(1)}$. However, we can avoid storing all the tuples in $\mathscr{D}$ at once. Instead, we can generate each tuple on-the-fly and run the \gst algorithm for that tuple before moving on to the next one. This way, we only need to store one tuple at a time, which requires $k\cdot n^{\Oh(1)}$ space. Therefore, the overall space complexity of the algorithm is $k\cdot n^{\Oh(1)}$.}
    \end{proof}
\csccds*
\begin{proof}
        Given an instance $(\calG,k)$ of \cds, we model it as an instance $(({\cal U, \cal F}),\hat{G})$ of \Csc as follows. We set $\cU = V(\calG)$. The family $\cF$ consists of the closed neighborhoods of all the vertices in $\calG$.
    Two sets $A,B \in \cF$ have an edge between them in $\hat{G}$ if the corresponding vertices share an edge in $\calG$. It is easy to show that the instance $(\calG,k)$ of \cds has a solution of size at most $k$ if and only if the $(({\cal U, \cal F}),\hat{G})$ instance of \Csc has a solution of size at most $k$.
\end{proof}
\subsection{Details of \Cref{rem:compact-size-1}}

\begin{lemma}\label{lem:compact-size-1}
The family $\mathscr{D}=\{\T^{\star}\}$ is a compact representation (in the sense of Definition~\ref{def:compact_representations})
of all inclusion-wise minimal set covers of size at most $k+1$ for the instance $(\cU,\cF,k+1)$.
In particular, $|\mathscr{D}|=1$.
\end{lemma}
\begin{proof}
We verify the three conditions in Definition~\ref{def:compact_representations} for the single tuple
$\T^{\star}=(F_0,F_1,\ldots,F_k)$.

\smallskip
\noindent\emph{Covering condition.}
Fix any choice $(f_0,f_1,\ldots,f_k)\in F_0\times F_1\times\cdots\times F_k$.
Then $f_0=B$. Moreover, for each $i\in[k]$, we have $f_i=S_{i,j_i}$ for some $j_i\in[m]$, hence $u_i\in f_i$.
Therefore $B$ covers $t$ and all elements $b_{i,j}$, while $\{f_i:i\in[k]\}$ covers all $u_1,\ldots,u_k$.
Thus $\bigcup_{i=0}^k f_i=\cU$.

\smallskip
\noindent\emph{Structural condition.}
The subfamilies $F_0,F_1,\ldots,F_k$ are pairwise disjoint as families of sets: $F_0$ contains only $B$, and for
each $i\in[k]$ the family $F_i$ consists only of sets $S_{i,j}$ with fixed first index $i$.

\smallskip
\noindent\emph{Completeness condition.}
Let $\cS\subseteq \cF$ be an inclusion-wise minimal set cover with $|\cS|\le k+1$.
By Lemma~\ref{lem:many-minimal} (or by the same argument), every inclusion-wise minimal set cover in this instance has
the form $\cS=\{B\}\cup\{S_{1,j_1},\ldots,S_{k,j_k}\}$,
for some $(j_1,\ldots,j_k)\in[m]^k$, and hence $|\cS|=k+1$.
Order $\cS$ as $(B,S_{1,j_1},\ldots,S_{k,j_k})$. Then $B\in F_0$ and $S_{i,j_i}\in F_i$ for each $i\in[k]$,
so $\cS$ is captured by $\T^{\star}$.

This establishes all three conditions. Hence $\mathscr{D}$ is a compact representation and has size $1$.
\end{proof}

\section{Missing proofs and details from \Cref{sec:cores}}
\label{sec:app-cores}
\existencecores*
\begin{proof}
The following result will be used to argue the base case in the proof of \Cref{thm:existence-domination-dominator-core}.
As noted in \Cref{sec:prelims}, we handle graphs with a dominating set of size at most one separately.
Thus, we shall assume that the size of a minimum connected dominating set of $G$ is strictly greater than $1$.

\begin{proposition}[Theorem 2.4 of \cite{DBLP:journals/toct/Guillemot25}]
    \label{prop : bounding number of neighbourhoods}
    Let $\mathcal{F}$ be a \dslfree family of sets over a universe of size $n$ such that no element of the universe appears in all the sets of $\mathcal{F}$. Then, the number of distinct sets in $\mathcal{F}$ is bounded by $\Oh(n^d)$. 
\end{proposition}
The \emph{image set} of a function $h: A \to B$ denotes the subset of elements of $B$ that have a pre-image in $A$, and is denoted by $f(A)$.

Using the reasoning of \cite{DBLP:journals/toct/Guillemot25}, we can also show the following observation, which will be used in the proof of base case of \Cref{lemma:inductive-core}, the main technical result of this section. For the sake of completeness, we provide the proof of this observation here.

\begin{observation}[\cite{DBLP:journals/toct/Guillemot25}]
\label{obs:existance of grouping}
    For the \dslfree set family $\Co{F}$, let $\intmath(\mathcal{F})$ denote the intersection closure of $\mathcal{F}$. If there exists a set $S' \in \intmath(\mathcal{F})$ with $k^{\ell(S')} < |S'|$, then Line~\ref{base case of induction algo: findcore} of \textsc{FindCore}$( \mathcal{U}, \mathcal{F}, f, g, S)$ is never executed.
\end{observation}
\ma{Deleted the proof here. As the reviewer points out}

The following lemma serves as the base case for the inductive proof of \Cref{thm:existence-domination-dominator-core}. It establishes that if the universe $\mathcal{U}$ is sufficiently small, then the image sets of $f$ and $g$ already form the required $k$-dominator and $k$-grouped domination cores, respectively.





\begin{lemma}\label{lemma:core base case}Consider a relevant tuple $(\U, \F, f,g)$ for a graph $G$.
    Suppose that the set family $\F$ is \dslfree.\hide{ with the corresponding universe $\U$.} Moreover, suppose that the image set of $f$ forms a $k$-dominator core of $G$, the image set of $g$ forms a $k$-grouped domination core of $G$, and $|\mathcal{U}| \le k^d$. Then, \textsc{FindCore}$(\mathcal{U}, \mathcal{F}, f, g)$ outputs a $k$-grouped domination core of size at most $k^d$ and a $k$-dominator core of size at most $\Oh(k^{d^2})$.
    \hide{\begin{itemize}
        \item \textsc{FindCore}$(\mathcal{U},\mathcal{F},  f, g)$ outputs a $k$-grouped domination core of size at most $k^d$; and
        \item \textsc{FindCore}$(\mathcal{U}, \mathcal{F}, f, g)$ outputs a $k$-dominator core of size at most $\Oh(k^{d^2})$.
    \end{itemize}}
\end{lemma}
\begin{proof}
    The statement of \Cref{lemma:core base case} follows directly from \Cref{prop : bounding number of neighbourhoods}, together with the assumption that the image of $f$ forms a $k$-dominator core and the image of $g$ forms a
    $k$-grouped domination core.
\end{proof}

We begin by defining the problem of \rbds, which will be used as a tool in the proof of the main lemma of this section. In \rbds, we are given a bipartite graph $G = (R \uplus B, E)$ and an integer $k$, and the goal is to determine whether there exists a subset $S \subseteq R$ of size at most $k$ such that every vertex in $B$ is adjacent to at least one vertex in $S$. We will call $R$ as the red set and $B$ as the blue set of $G$.


Let $\hat{G}= (R\uplus B, E)$ be the graph constructed as follows. We create two copies of the vertex set $V(G)$, denoted by $R$ and $B$. We add edge $(x,y)$ to $E$ if the corresponding vertex of $x$ in $R$ is adjacent to the corresponding vertex of $y$ in $B$ in the original graph $G$ or if $x$ and $y$ correspond to the same vertex in $G$.

Now, for a relevant tuple $( \U, \F, f, g, S)$, we define the \emph{relevant red-blue dominating set} instance as follows.
$V(G') = R \uplus B $ where $R$ and $B$ are copies of vertices in $V(G)$. Moreover, the edges $E(G')$ are defined as follows: for each subset $S'\in \mathcal{F}$, we consider each element $v \in f(S')$. The corresponding vertex $v\in R$ is adjacent to the corresponding vertices of $B$ in the set $\cup_{s\in S'}g(s) $.
Both of these constructions will be used to design an invariant which will be useful to prove the existence of a bounded $k$-dominator core and $k$-grouped domination core and thus help in proving the correctness using strong induction. One will give an handle to the original graph $G$ and the other will give a handle to the properties of the relevant tuple $( \U, \F, f, g)$.




\begin{observation}\label{obs:rbds-ds}
    A set $S \subseteq V(G)$ is a dominating set of $G$ if and only if the corresponding vertices of $S$ in $R$ form a red-blue dominating set of $\hat{G}$.
\end{observation}
This preserves the dominating sets of $G$. In particular, it preserves all the minimal dominating sets of $G$. This is an observation that will be used in the following pivotal lemma.


\begin{restatable}[\app]{lemma}{inductivecore}
\label{lemma:inductive-core}
Let $(\mathcal{U},\mathcal{F},f,g)$ be a relevant tuple of a graph $G$. Suppose that:
\begin{enumerate}
    \item $\mathcal F$ is $d$-semi-ladder-free;
    \item the image of $f$ is a $k$-dominator core of $G$;
    \item the image of $g$ is a $k$-grouped domination core of $G$; and
    \item for every $D\subseteq V(G)$ with $|D|\le k$, $D$ dominates $G$ if and only if the corresponding red vertices dominate the current red-blue instance $G'(\mathcal F,\mathcal U,f,g)$.
\end{enumerate}
Then \textsc{FindCore} outputs a $k$-grouped domination core of size at most $k^d$ and a $k$-dominator core of size at most $\Oh(k^{d^2})$.
\end{restatable}

\begin{proof}
We prove the lemma by induction on $|\mathcal U|$.
If $|\mathcal U|\le k^d$, then the image of $g$ already consists of at most $k^d$ batches. Moreover, by the set-counting bound of Guillemot (\Cref{prop : bounding number of neighbourhoods}), the number of distinct sets in the image of $f$ is $\Oh(k^{d^2})$. Hence the algorithm outputs a $k$-grouped domination core of size at most $k^d$ and a $k$-dominator core of size at most $\Oh(k^{d^2})$.
Assume that the statement holds for every relevant tuple
$(\mathcal U,\mathcal F,f,g)$ satisfying the four assumptions with
$|\mathcal U|<t$. A vertex $v$ is said to \emph{partially dominate} a set of vertices $B$ if
$v$ does not dominate all
vertices of $B$.

\Ma{Consider an instance with $|\mathcal U|=t>k^d$.
If there is no set $S\in\intmath(\mathcal F)$ satisfying
$|S|>k^{\ell(S)}$,
then $|\mathcal U| \le k^d$ as $U$ is also in $\intmath(\mathcal F)$. Hence, we arrive at a contradiction. Therefore, there exists a set $S\in\intmath(\mathcal F)$ satisfying $|S|>k^{\ell(S)}$.} 
Let
$S=\{v_{i_1},\ldots,v_{i_q}\}\in\intmath(\mathcal F)$ be one with
minimum possible value $\ell(S)=i$.
We first show that every dominating set of the current red-blue
instance of size at most $k$ contains a vertex that dominates the
entire batch
$g(S)=\bigcup_{v\in S}g(v)$.

Indeed, by the minimality of $i$, every member of $\mathcal F$
intersects $S$ in at most $k^{i-1}$ elements unless it contains all of
$S$. Consequently, every red vertex that does not dominate the entire
batch $g(S)$ dominates vertices from at most $k^{i-1}$ batches indexed
by elements of $S$. Since $|S|>k^i$, at most $k$ such vertices cannot
collectively dominate every batch represented by $S$. Hence every
dominating set of size at most $k$ must contain a vertex adjacent to
every vertex of $g(S)$. \Ma{This proof idea is similar to the one used in Lemma 3.4 of \cite{DBLP:journals/toct/Guillemot25}.}

The algorithm now applies the \textsc{Group} operation and obtains a
new relevant tuple
$(\mathcal U',\mathcal F',f',g')$.
We verify that the recursive instance again satisfies the assumptions
of the lemma.
First, $\mathcal F'$ is $d$-semi-ladder-free by
\Cref{prop:semiladder closed grouping}.




Next, we show that the image of $f'$ is a $k$-dominator core. We verify the three properties of \Cref{def : Dominator Core}.
For the enclosure property, observe that
$
\bigcup_{A\in\mathcal F'}f'(A)
=
\bigcup_{A\in\mathcal F}f(A)$,
since each image set of $f'$ is obtained by merging one or more image
sets of $f$. Hence no vertex belonging to the original dominator core
is removed during the grouping operation. As the image of $f$ is a
$k$-dominator core by assumption, every inclusion-wise minimal
dominating set of size at most $k$ remains contained in the union of
the image of $f'$.
To prove the unique contribution property, let $D$ be an inclusion-wise
minimal dominating set of size at most $k$. Suppose that $D$ contains
two vertices $u$ and $v$ from the same image set of $f'$. By
construction of $f'$, the vertices $u$ and $v$ have identical
neighborhoods in the grouped red-blue instance, since the only edges
deleted during the grouping operation correspond to partial domination
of the newly created batch. Consequently, $u$ and $v$ dominate exactly
the same blue vertices. Removing either $u$ or $v$ therefore preserves
domination of the grouped red-blue instance. By Assumption~4, the
corresponding set is still a dominating set of $G$, contradicting the
minimality of $D$. Hence every minimal dominating set contains at most
one vertex from each image set of $f'$.
Finally, let $u\in D\cap f'(A)$ and let $v\in f'(A)$ be arbitrary.
Again, $u$ and $v$ have identical neighborhoods in the grouped
red-blue instance. Therefore replacing $u$ by $v$ preserves domination
of every blue vertex. By Assumption~4, the corresponding set also
dominates $G$. Thus
$(D\setminus\{u\})\cup\{v\}$
is a dominating set of $G$, establishing the replacement property.
Hence the image of $f'$ is a $k$-dominator core.

We next show that the image of $g'$ is a $k$-grouped domination core.
Let $\mathcal G'$ denote the image of $g'$.
We first verify the collective domination property. Let $D$ be an
inclusion-wise minimal dominating set of $G$ of size at most $k$. By
Assumption~4, the corresponding red vertices form a dominating set of
the grouped red-blue instance.
Consider any batch $B\in\mathcal G'$. If $B$ is one of the batches that
was not modified by the grouping operation, then the property follows
directly from the assumption that the image of $g$ is a
$k$-grouped domination core. Otherwise, $B=g'(v_S)=\bigcup_{s\in S}g(s)$.
By the counting argument proved above, every dominating set of size at
most $k$ contains a vertex that dominates every vertex of $B$. Hence
every batch of $\mathcal G'$ satisfies the collective domination
property.
It remains to prove the preservation property. Let $D$ be any subset of
vertices that dominates every batch of $\mathcal G'$. Every unchanged
batch of the original image of $g$ is also a batch of $\mathcal G'$ and
is therefore dominated by $D$. Moreover, if
$B=g'(v_S)=\bigcup_{s\in S}g(s)$, then dominating $B$ implies dominating
each original batch $g(s)$ for every $s\in S$, since each of them is a
subset of $B$. Therefore $D$ dominates every batch of the image of $g$.
As the image of $g$ is a $k$-grouped domination core, it follows that
$D$ is a dominating set of $G$.
Hence the image of $g'$ is a $k$-grouped domination core.

Finally, let $D$ be any subset of red vertices of size at most $k$.
Every feasible dominating set already contains a vertex adjacent to the
entire batch $g(S)$, so deleting edges corresponding to partial
domination cannot affect whether $D$ dominates the grouped red-blue
instance. Hence $D$ dominates the original red-blue instance if and
only if it dominates the grouped one. Together with Assumption~4, this
shows that the fourth assumption also holds for the recursive
instance.

Thus all four assumptions are satisfied by
$(\mathcal U',\mathcal F',f',g')$. Since
$|\mathcal U'|<|\mathcal U|$, the induction hypothesis applies.
Therefore the recursive call returns a $k$-grouped domination core of
size at most $k^d$ and a $k$-dominator core of size at most
$\Oh(k^{d^2})$. The same conclusion therefore holds for the original
instance, completing the proof.
\end{proof}

Armed with \Cref{lemma:inductive-core}, we now invoke \textsc{FindCore} on the original graph.
    We initialize the relevant tuple $(\mathcal{U},\mathcal{F},f,g)$ as follows: 
    \[ \U=V(G), \quad \F = \{F_v=N_G[v] \mid v\in V(G)\}, \quad f(F_v)=\{v\}, \textrm{ and } g(v)=\{v\}.\]
    
    
    We verify that this initialization satisfies all four conditions of \Cref{lemma:inductive-core}:
    \begin{itemize}
        \item By \Cref{def:graph-semiladder}, the closed-neighborhood family $\mathcal F$ is \dslfree.  
        \item The image set of $f$, $f(\F)$, forms a $k$-dominator core of $G$ since it consists of singletons $\{v\}$ for each $v \in V(G)$; hence, every minimal dominating set is trivially contained in this family and intersects each set in at most one vertex.
        \item The image set of $g$ forms a $k$-grouped domination core of $G$ for the same reason, as any minimal dominating set dominates every singleton $\{v\}$.
        \item The auxiliary graph $G'(\F,\U,f,g)$ coincides with $\hat{G}$, since for every $S = N_G[v] \in \mathcal{F}$, all vertices in $f(S) = \{v\}$ are adjacent to all vertices in $g(v) = \{v\}$ by construction. Therefore, the red-blue invariant holds for all sets of size at most $k$.
    \end{itemize}
    
    Hence, all the preconditions of \Cref{lemma:inductive-core} hold for this instance.
    Applying \textsc{FindCore}$( \mathcal{U}, \mathcal{F}, f, g)$ yields the desired output. 
    This completes the proof of \Cref{thm:existence-domination-dominator-core}.
\end{proof}
\clearpage
\subsection{Pseudo-code of the \textsc{Group} operation}
\begin{algorithm}[H]
    \caption{\textsc{Group}$(\U, \F,  f, g, S)$}
    \label{alg:group}
    \begin{algorithmic}[1]
        \Require $\F \subseteq 2^{\U}$, functions $f : \F \to 2^{V(G)}$, $g : \U \to 2^{V(G)}$, and a subset $S \subseteq \U$.
        \State Create a new representative $v_S\notin \U$ and set $\U'\gets (\U\setminus S)\cup \{v_S\}$.
        \For{each $A\in \F$}
            \State $\sigma(A)\gets (A\setminus S)\cup\{v_S\}$ if $S\subseteq A$, and $\sigma(A)\gets A\setminus S$ otherwise.
        \EndFor
        \State $\F'\gets \{\sigma(A):A\in\F\}$, with duplicate images identified.
        \For{each $B\in\F'$}
            \State $f'(B)\gets \bigcup_{A\in\F:\sigma(A)=B} f(A)$.
        \EndFor
        \State $g'(u)\gets g(u)$ for every $u\in \U\setminus S$, and $g'(v_S)\gets \bigcup_{s\in S}g(s)$.
        \State \Return $(\U', \F',f',g')$.
    \end{algorithmic}
\end{algorithm}

\section{Missing proofs and details of \Cref{sec:lossy-kernel}}
\label{sec:app-lossy-kernel}
\additiveone*
{
    \begin{proof}
        Let $S$ be a set of vertices of a graph $H$.
        Let us denote $\bigcap_{v \in S}N_H[v]$ as $N_H[S]$ and $\bigcap_{v' \in \hg(v)}N_H[v']$ as $N_H[\hg(v)]$. Also, let us denote $\bigcap_{v \in S}N_H[\hg(v)]$ as $N_H[\hg(S)]$.
        
        Consider a chain $\mathcal{C}$ of length $l$ in $\intmath(\hat{G})$. By deleting all heavy vertices from each set in the chain, we show that the resulting sets belong to $\intmath(G)$, thereby proving the existence of a chain in $\intmath(G)$ of length $l-2$.

        To establish this, we examine the structural properties of the intersection closure. Any set in the closure can be viewed as the intersection of closed neighborhoods of a specific vertex subset. We define the functions
        $g_G : 2^{V(G)} \rightarrow 2^{V(G)}$ and $g_{\hat{G}} : 2^{V(\hat{G})} \rightarrow 2^{V(\hat{G})}$
        on the subsets of vertices for $G$ and $\hat{G}$, respectively. For a set $S$, let $g_G(S)$ denote an arbitrary set of vertices whose closed neighborhoods intersect to yield exactly $S$ in the graph $G$, and let $g_{\hat{G}}(S)$ denote an arbitrary set of vertices whose closed neighborhoods intersect to yield exactly $S$ in the graph $\hat{G}$. Under this formulation, we can verify the claim by leveraging the property that the heavy vertices form an independent set in ${\hat{G}}$.

        \begin{claim}
            \label{claim:containing-closures}
            For any set $S \in \intmath(\hat{G})$, the intersection $S \cap V(G)$ is an element of $\intmath(G)$.
        \end{claim}

        \begin{claimproof}
            Suppose $S = N_{\hG}[X]$ for some vertex set $X \subseteq V(\hat{G})$. We observe that $S \cap V(G) = N_G[\hg(X)]$, where $\hg(X)$ is the set of vertices in $G$ obtained by replacing each heavy vertex in $X$ with its corresponding vertex set in $G$. The intersection of neighborhoods in $G$ of this expanded set $\hg(X)$ yields exactly $S \cap V(G)$. Since $S \cap V(G)$ is expressed as an intersection of closed neighborhoods in $G$, it follows that $S \cap V(G) \in \intmath(G)$.
        \end{claimproof}
        Let us examine two consecutive sets $S_1$ and $S_2$ in the chain $C$ of $\intmath(\hat{G})$ such that $S_1 \supset S_2$. We say $u$ is in the heavy vertex $s$ if $\hg(s)$ contains $u$.
        \begin{claim}
            If $S_2$ contains at least two heavy vertices, then a heavy vertex $s$ is an element of $S_2$ if and only if $\hg(s) \subseteq S_2$.
        \end{claim}

        \begin{claimproof}
            Since $S_2$ contains at least two heavy vertices, we first observe that $g_{\hat{G}}(S_2)$ cannot contain any heavy vertices. This is because heavy vertices form an independent set in $\hat{G}$, and a heavy vertex $s$ is only contained in its either its own closed neighborhood or in that of a non-heavy vertex. Therefore, we must have $g_{\hat{G}}(S_2) \subseteq V(G)$.

            Recall that $S$ is defined as the intersection of the closed neighborhoods of vertices in $g_{\hat{G}}(S_2)$. A heavy vertex $s$ belongs to $S_2$ if and only if $s$ is adjacent to every vertex in $g_{\hat{G}}(S_2)$. By the construction of $\hat{G}$, a non-heavy vertex $v$ is adjacent to a heavy vertex $s$ if and only if $v$ is adjacent to every vertex in the set $\hg(s)$. Consequently, $s \in S_2$ if and only if every $v \in g_{\hat{G}}(S_2)$ is adjacent to every $u \in \hg(s)$, which is equivalent to the condition $\hg(s) \subseteq S_2$.
        \end{claimproof}

            
        The above claim establishes that if there are atleast two heavy vertices in $S_2$ and if a heavy vertex belongs to $S_1$ but
        not to $S_2$, then at least one vertex of $\hg(s)$ belongs to
        $S_1\setminus S_2$. Consequently,
        $S_1\cap V(G)\neq S_2\cap V(G)$ whenever $S_2$ contains at least two
        heavy vertices. 
        
        By \Cref{claim:containing-closures}, both
        $S_1\cap V(G)$ and $S_2\cap V(G)$ belong to $\intmath(G)$.
        If neither $S_1$ nor $S_2$ contains a heavy vertex, then
        $S_1\cap V(G)=S_1$ and $S_2\cap V(G)=S_2$, so they are again distinct
        members of $\intmath(G)$.
        If $S_1\cap V(G) \neq S_2\cap V(G)$, then they form two different sets in the $\intmath(G)$ and $S_1\cap V(G) \supset S_2\cap V(G)$ as $S_1 \supset S_2$ and none of the heavy vertices are present in $V(G)$.

        It therefore remains to consider the case where $S_2$ contains at most one heavy vertex and $S_1 \cap V(G) = S_2 \cap V(G)$.
        Let us first consider pairs of consecutive sets in the chain where $S_2$ contains exactly one heavy vertex, and $S_1 \cap V(G) = S_2 \cap V(G)$.
        Observe that the number of heavy vertices in $S_1$ is strictly more than that in $S_2$.
        Thus, such a pair can occur only once in the entire chain. 
        Next, let us consider pairs of consecutive sets in the chain where $S_2$ contains no heavy vertex, and $S_1 \cap V(G) = S_2 \cap V(G)$.
        Again, such a pair can occur only once in the entire chain. 
        
        Once all heavy vertices have been removed from the chain, every set is in $\intmath(G)$, and by our above analysis, at least $l-1$ of them are pairwise distinct. 
        Therefore, deleting the heavy vertices from every set in the chain
        produces a chain in $\intmath(G)$ whose length is at least $l-2$. Since
        $\intmath(G)$ is $d$-semi-ladder-free, every chain in
        $\intmath(G)$ has length at most $d$. Hence
        $
        l-2\le d
        $,
        or equivalently,
        $
        l\le d+2
        $.
        By \Cref{prop : bounded_length_of_int_closure}, every semi-ladder
        has length at most the maximum chain length. Hence $\hat G$ has
        semi-ladder index at most $d+2$, proving the lemma.

    \end{proof}
}
\lossycompression*
\begin{proof}
    {

The following propositions will help us in our kernel.
\begin{definition}[\cite{DBLP:journals/siamdm/EibenKMPS19}]
    Let $D$ be a connected graph and $t\in {\mathbb N}$.  A
    \emph{$(D,t)$-covering family} is a family $\mathcal{F}(D,t)$ of
    connected subgraphs of $D$ such that $(i)$ for each
    $T \in {\mathcal{F}(D,t)}$, $\vert V(T) \vert \leq 2t$ and $(ii)$
    $\bigcup_{T \in \mathcal{F}(D,t)}{V(T)} = V(D)$.
\end{definition}

\begin{proposition}[\cite{DBLP:journals/siamdm/EibenKMPS19}]\label{prop:covering-family}
    Let $D$ be a connected graph and $t\in {\mathbb N}$.  Then there is
    a $(D,t)$-covering family $\mathcal{F}(D,t)$ such that
    $|\mathcal{F}(D,t)| \leq \frac{|V(D)|}{t} + 1$, and
    $\sum_{T \in \mathcal{F}(D,t)}{|V(T)|} \leq (1 + \frac{1}{t}) |V(D)|
        + 1$.
\end{proposition}

By leveraging the structural properties of $\hat{G}$ and the preceding propositions, we establish that \cds for \dslfree graphs admits a polynomial $(1+\epsilon)$-lossy compression to \cds for \deslfree graphs.
}
\solnlifting*
{
    \begin{proof}
        We repeatedly remove heavy vertices from the solution.  Let $h=v_B\in S'\cap H$ be a heavy vertex.  If $|S'|=1$, then the existence of the connected dominating set $\{h\}$ implies that some vertex of $G$ dominates the batch $B$; replacing $h$ by such a vertex gives a connected dominating set of $G$ of size one.  Assume therefore that $|S'|>1$.

        Since $G'[S']$ is connected, $h$ has a neighbor $u\in S'\setminus\{h\}$.  By construction of $\hat G$, every non-heavy neighbor of $h$ is adjacent in $G$ to every vertex of the batch $B$.  Choose any $b\in B$ and set $S''=(S'\setminus\{h\})\cup\{b\}$.  The vertex $b$ is adjacent to every selected neighbor of $h$, so replacing $h$ by $b$ preserves connectivity.  Moreover, every original vertex whose domination was certified through $h$ is dominated by $b$.  Thus $S''$ is a connected dominating set with one fewer heavy vertex and $|S''|\le |S'|$.  Iterating this replacement yields a connected dominating set contained in $V(G)$; by the preservation property of the $k$-grouped domination core, it dominates $G$.
    \end{proof}
}
\lossycompsizeapprox*
{
    \begin{proof}
        The size bound follows from the core bounds.  We have $|H\cup NH|\le |\G|\le k^d$.  The algorithm marks only trees for subsets of the dominator core of size at most $2t$, and each such tree has at most $2t$ vertices.  Since $|\C|=\Oh(k^{d^2})$, we get $|M|\le |\C|^{2t}\cdot 2t=k^{\Oh(td^2)}$, and $|M\cup C_{\rm conn}|\le 3|M|$.  Hence $|V(G')|\le k^d+k^{\Oh(td^2)}$.  With $t=\Theta(1/\epsilon)$ this is $k^{\Oh(d^2/\epsilon)}$.

        We next prove the approximation guarantee.  Let $L=\lceil 8/\epsilon\rceil$.  As a preprocessing step, the compression algorithm first checks by brute force whether $G$ has a connected dominating set of size at most $L$.  If it does, it stores such an optimum small solution and outputs a constant-size equivalent instance whose lifting algorithm returns it.  Thus assume below that $\operatorname{OPT}(I,k)>L$.

        Let $D^\star$ be an optimum connected dominating set of $G$ with $|D^\star|\le k$; if no such set exists, then the claim is immediate from the definition of the thresholded objective.  Let $D_0\subseteq D^\star$ be an inclusion-wise minimal dominating set.  Consider a $(D^\star,t)$-covering family as in \Cref{prop:covering-family}.  For each connected subgraph $T$ in this family, let $Q_T$ be the set of dominator-core classes that contain vertices of $T\cap D_0$.  The unique-contribution property of the dominator core implies $|Q_T|\le |V(T)|\le 2t$, so the compression algorithm marks a tree of size at most $2t$ hitting all groups in $Q_T$.  Taking the union of these marked trees gives a set that dominates all batches of the $k$-grouped domination core, and hence dominates $G$.

        The covering-family bound gives total marked size at most $(1+1/t)|D^\star|+1$.  These marked vertices induce at most $|D^\star|/t+1$ connected components, so \Cref{prop:approx-cds-by-ds} adds at most $2|D^\star|/t+2$ connector vertices.  Therefore
        \[
        \operatorname{OPT}(I',k)\le (1+3/t)|D^\star|+3.
        \]
        Choose $t\ge 8/\epsilon$.  Since $|D^\star|>8/\epsilon$, the additive term satisfies $3\le (\epsilon/2)|D^\star|$, and $3|D^\star|/t\le (3\epsilon/8)|D^\star|$.  Hence $\operatorname{OPT}(I',k)\le (1+\epsilon)|D^\star|=(1+\epsilon)\operatorname{OPT}(I,k)$.
    \end{proof}
}

    By \Cref{lemma:soln-lifting}, we have $CDS(G,k,S) \le CDS(G',k,S)$.
By \Cref{lem:lossy-comp-size-approx}, we have ${\rm OPT}(I',k) \le (1+\varepsilon){\rm OPT}(I,k)$.
Therefore, $$\frac{CDS(G,k,S)}{{\rm OPT}(I,k)} \le \frac{CDS(G',k,S)}{\frac{1}{1+\epsilon}{\rm OPT}(I',k)} \le (1+\epsilon) \frac{CDS(G',k,S)}{{\rm OPT}(I',k)}.$$
Combining this with the solution lifting algorithm in \Cref{lemma:soln-lifting} and the above approximation guarantee completes the proof of \Cref{thm: lossy_compression}.
\end{proof}

\end{document}